%% file: main.tex
\documentclass[11pt]{article}
\usepackage[margin=1in]{geometry}

\input{preamble}

\title{New Philosopher Inequalities for Online Bayesian Matching,\\via Pivotal Sampling}
\author[1]{Mark Braverman\thanks{Supported in part by the NSF Alan T. Waterman Award, Grant No. 1933331.}}
\author[2]{Mahsa Derakhshan}
\author[3]{Tristan Pollner}
\author[3]{Amin Saberi\thanks{Supported in part by  NSF Awards CCF2209520, CCF2312156, and a gift from CISCO.}}
\author[4]{David Wajc\thanks{Supported in part by a Taub Family “Leaders in Science \& Technology” Fellowship. Work done in part while the author was at Stanford University and Google Research.}}
\affil[1]{Princeton University}
\affil[2]{Northeastern University}
\affil[3]{Stanford University}
\affil[4]{Technion}

\begin{document}

\include{xspace}

\pagenumbering{gobble}
\date{}
\maketitle

\begin{abstract}  

    We study the polynomial-time approximability of the optimal online stochastic bipartite matching algorithm, initiated by Papadimitriou et al.~(EC'21).
    Here, nodes on one side of the graph are given upfront, while at each time $t$, an online node and its edge weights are drawn from a time-dependent distribution. The optimal algorithm is $\PSPACE$-hard to approximate within some universal constant. 
    We refer to this optimal algorithm, which requires time to think (compute), as a \emph{philosopher}, and refer to polynomial-time online approximations of 
    the above as \emph{philosopher inequalities}.
    The best known philosopher inequality for online matching yields a $0.652$-approximation. In contrast, the best possible prophet inequality, or approximation of the optimum offline solution, is $0.5$.
    
    Our main results are a $\edgeapprox$-approximate algorithm and a $\vertexapprox$-approximation for a vertex-weighted special case.
    Notably, both bounds exceed the $0.666$-approximation of the offline optimum obtained by Tang, Wu, and Wu (STOC'22) for the vertex-weighted problem. Building on our algorithms and the recent black-box reduction of Banihashem et al.~(SODA'24), we provide polytime (pricing-based) truthful mechanisms which $0.678$-approximate the social welfare of the optimal online allocation for bipartite matching markets.

    Our online allocation algorithm relies on the classic pivotal sampling algorithm (Srinivasan FOCS'01, Gandhi et al.~J.ACM'06), along with careful discarding to obtain strong negative correlations between offline nodes, while matching using the highest-value edges. Consequently, the analysis boils down to examining the distribution of a weighted sum $X$ of negatively correlated Bernoulli variables, specifically lower bounding its mass below a threshold, $\E[\min(1,X)]$, of possible independent interest. 
    Interestingly, our bound relies on an \emph{imaginary} invocation of pivotal sampling.

\end{abstract}

\newpage
\pagenumbering{arabic}
\input{new-intro}

\input{prelims}

\input{algo}

\input{Emin1X}

\input{edgeweighted}

\input{vertexweighted}

\input{hardness}

\bibliographystyle{alpha}
\bibliography{abb,a,ultimate}

\appendix

\section*{APPENDIX}

\input{app-prelims}

\input{app-algo}

\input{app-Emin1X}
\input{app-edgeweighted}
\input{app-vertexweighted}

\input{app-hardness}

\input{non-bernoulli}

\input{app-pricing}

\end{document}

%% file: preamble.tex
\usepackage{booktabs} 

\usepackage[utf8]{inputenc}
\usepackage{authblk}
\usepackage{amsmath, amssymb, amsthm, thmtools, amsfonts, bm, bbm, thm-restate}
\usepackage{enumitem}
\usepackage{cite}
\usepackage[normalem]{ulem}
\renewcommand{\vec}[1]{\mathbf{#1}}

\allowdisplaybreaks

\usepackage{asymptote}
\usepackage{graphicx}
\usepackage{todonotes}
\usepackage{multirow}
\usepackage{comment}
\usepackage{dsfont}
\usepackage{bigstrut}
\usepackage{caption}
\usepackage{subcaption}
\usepackage{tcolorbox}
\usepackage{algorithm}
\usepackage[noend]{algpseudocode}
\usepackage{graphicx}

\usepackage{xcolor}
\definecolor{BrickRed}{rgb}{0.8,0.25,0.33}
\usepackage[pagebackref,colorlinks,citecolor=blue,linkcolor=BrickRed]
{hyperref}

\usepackage[framemethod=tikz]{mdframed}
\usepackage{nicefrac}

\usepackage{pifont} 

\usepackage{cleveref}

\theoremstyle{definition}
\newtheorem{example}{Example}[]
\theoremstyle{plain}
\newtheorem{thm}{Theorem}[section]

\newtheorem{prop}[thm]{Proposition}

\newtheorem{fact}[thm]{Fact}
\newtheorem{lem}[thm]{Lemma}
\newtheorem{lemma}[thm]{Lemma}

\newtheorem{claim}[thm]{Claim}
\newtheorem{Def}[thm]{Definition}
\newtheorem{obs}[thm]{Observation}

\newtheorem{rem}[thm]{Remark}

\theoremstyle{definition}
\newtheorem{definition}{Definition}[section]

\crefname{thm}{Theorem}{theorems}
\crefname{cla}{Claim}{claims}
\crefname{lem}{Lemma}{lemmas}
\crefname{fact}{Fact}{facts}

\newcommand{\E}{\mathbb{E}}

\newcommand{\Std}{\textrm{Std}}
\newcommand{\Ber}{\mathop{\textrm{Ber}}}
\newcommand{\Bin}{\mathop{\textrm{Bin}}}

\newcommand{\eps}{\varepsilon}
\newcommand{\calA}{\mathcal{A}}
\newcommand{\calB}{\mathcal{B}}
\newcommand{\calD}{\mathcal{D}}

\newcommand{\calH}{\mathcal{H}}
\newcommand{\calI}{\mathcal{I}}
\newcommand{\calM}{\mathcal{M}}

\newcommand{\OFF}{[n]}
\newcommand{\Var}{\textup{Var}}

\newcommand{\bigmid}{\,\,\middle\vert\,\,}
\newcommand{\conv}{\textup{\textsf{conv}}}

\newcommand{\opton}{\mathrm{OPT}_\mathrm{on}}

\newenvironment{wrapper}[1]
{
	\smallskip
	\begin{center}
		\begin{minipage}{\linewidth}
			\begin{mdframed}[hidealllines=true, backgroundcolor=gray!20, leftmargin=0cm,innerleftmargin=0.35cm,innerrightmargin=0.35cm,innertopmargin=0.375cm,innerbottommargin=0.375cm,roundcorner=10pt]
				#1}
			{\end{mdframed}
		\end{minipage}
	\end{center}
	\smallskip
}

\newcommand*{\todos}{}%

\usepackage[T1]{fontenc}
\usepackage{lmodern}

\ifdefined\notodos
\def\tristan#1{}
\def\david#1{}
\def\mahsa#1{}
\def\mark#1{}
\def\amin#1{}
\fi

\ifdefined\todos
\def\tristan#1{\marginpar{$\leftarrow$\fbox{T}}\footnote{$\Rightarrow$~{\sf\textcolor{cyan}{#1 --Tristan}}}}
\def\david#1{\marginpar{$\leftarrow$\fbox{D}}\footnote{$\Rightarrow$~{\sf\textcolor{blue}{#1 --David}}}}
\def\mahsa#1{\marginpar{$\leftarrow$\fbox{MD}}\footnote{$\Rightarrow$~{\sf\textcolor{purple}{#1 --Mahsa}}}}
\def\mark#1{\marginpar{$\leftarrow$\fbox{MB}}\footnote{$\Rightarrow$~{\sf\textcolor{green}{#1 --Mark}}}}
\def\amin#1{\marginpar{$\leftarrow$\fbox{A}}\footnote{$\Rightarrow$~{\sf\textcolor{red}{#1 --Amin}}}}
\fi

\newcommand{\edgeapprox}{0.678}
\newcommand{\vertexapprox}{0.685}

\newcommand{\acoeff}{0.614}
\newcommand{\bcoeff}{0.122}
\newcommand{\ccoeff}{0.197}
\newcommand{\PSPACE}{\mathsf{PSPACE}}

\usepackage{float}
\usepackage{wrapfig}
\restylefloat{figure}

%% file: new-intro.tex
\section{Introduction}

We consider the \emph{online stochastic bipartite matching} problem. Here,
nodes of one side of a bipartite graph (offline nodes) are given up front; at timestep $t$, a node $t$ on the opposite side of the graph (an online node) reveals the weights of its edges to offline nodes, drawn from a time-dependent distribution. 
For example, in the ``Bernoulli case'' each 
online node $t$ arrives with known probability $p_t$ with known edge weights, otherwise all its edges' weights are zero.
Upon the realization of an online node $t$'s incident edge-weights, the algorithm must make an immediate and irreversible decision on whether and how to match $t$, with the objective of maximizing the weight of the resulting matching. 

When offline nodes correspond to items to sell and online nodes to impatient buyers, this problem is reminiscent of multidimensional auctions with unit-demand buyers, i.e., bipartite matching markets. 
Indeed, this connection is not only syntactic: 
\emph{pricing-based} algorithms obtaining high value for the above problem imply truthful mechanisms that approximate the optimal social welfare for such markets \cite{hajiaghayi2007automated,chawla2010multi,kleinberg2019matroid}.

The online stochastic bipartite matching problem has been intensely studied over the years via the lens of competitive analysis, or \emph{prophet~inequalities}, so-called as they compare with a ``prophet'' who has foresight of the realized graph, and computes the optimal allocation up front. A competitive ratio of $0.5$ is known to be achievable via myriad approaches \cite{feldman2015combinatorial,ezra2020online,dutting2020prophet},
 including pricing-based ones \cite{feldman2015combinatorial, dutting2020prophet},
and this is tight even without restricting to pricing-based algorithms, and even with a single offline node, or single item to sell \cite{krengel1978semiamarts}. 
For the unweighted and vertex-weighted problem, online (non-stochastic) bipartite matching algorithms yield a competitive ratio of $1-1/e\approx 0.632$ \cite{karp1990optimal,aggarwal2011online} (see also \cite{devanur2013randomized,eden2021economics}),
while the stochastic information allows for a better ratio of  $0.666$, but no more than $0.75$ \cite{tang2022fractional}. The special case of i.i.d.~distributions at all time steps has also been extensively studied \cite{feldman2009online,bahmani2010improved,mahdian2011online,karande2011online,haeupler2011online,manshadi2012online,jaillet2013online,brubach2016new,huang2019online,huang2021online,huang2022power,tang2022fractional,yan2024edge}.

Recently, Papadimitriou et al.~\cite{papadimitriou2021online} pioneered the study of the online stochastic bipartite matching problem via the lens of polytime (online) approximation of the optimal \emph{online algorithm}. The optimal online algorithm lacks foresight about the future, but possesses unlimited computational power, or time to ``think.'' We fittingly call this optimal policy a ``philosopher,'' and refer to approximation of this policy's value by polytime online algorithms as \emph{philosopher inequalities}.\footnote{We introduce this nomenclature here for the first time, although following talks by the authors at various venues, follow-up work \cite{dehaan2024matroid} and the survey \cite{huang2024online} have since adopted it.}
\cite{papadimitriou2021online} showed that unless $\mathsf{P}=\mathsf{PSPACE}$, no $(1-\epsilon)$-approximate philosopher inequality is possible for some constant $\epsilon > 0$, already for the Bernoulli problem.\footnote{In contrast, for a single offline node, with randomly permuted distributions (i.e., the prophet secretary problem), a $(1-\eps)$-approximate philosopher inequality is possible for any constant $\eps>0$ \cite{dutting2023prophet}.}
Thus, the philosopher is aptly named, as they require time to think in order to obtain optimal guarantees.\footnote{The modern equivalent of philosophers, i.e., professors, should stress this point to research funding agencies.}
In contrast, \cite{papadimitriou2021online} showed that the best philosopher inequality yields a higher approximation ratio than the best-possible prophet inequality's competitive ratio of $0.5$.
The $0.51$ bound  was subsequently improved \cite{saberi2021greedy,braverman2022max,naor2023online}, with the current best known approximation being $0.652$ \cite{naor2023online}.

Besides allowing for better quantitative bounds than achievable versus the offline optimal, the study of approximation of the optimal online algorithm opens avenues for more algorithmic and analytic ideas for online resource allocation problems. Mechanism design considerations furthermore beg the question of what guarantees for this problem can be had while maintaining incentive compatibility.

\subsection{Our Contributions}
Our main result is a new state-of-the-art philosopher inequality for online stochastic bipartite matching.
\begin{wrapper}
\begin{restatable}{thm}{thmedgeweightedintro}\label{thm-intro:edge-weighted}
(See Theorems \ref{thm:edge-weighted} and \ref{thm:edge-weighted-generalization}) There exists a polynomial-time $\edgeapprox$-approximate online algorithm for edge-weighted online stochastic bipartite matching.   
\end{restatable}
\end{wrapper}

Our techniques also yield a significantly simpler proof of the $(1-1/e)$ bound due to \cite{braverman2022max} (see \Cref{claim:bdmlsimplification}).
Besides improving the $0.652$ approximation ratio of \cite{naor2023online} for our problem, our bound of \Cref{thm-intro:edge-weighted} also exceeds the best competitive ratio of $0.666$ for the \emph{vertex-weighted} version of this problem \cite{tang2022fractional}.  Our second result is a better philosopher inequality for the vertex-weighted Bernoulli problem.

\begin{wrapper}
\begin{restatable}{thm}{thmvertexweightedintro}\label{thm-intro:vertex-weighted}
(See \Cref{thm:vertex-weighted}) There exists a polynomial-time $\vertexapprox$-approximate online algorithm for vertex-weighted online Bernoulli bipartite matching.   
\end{restatable}
\end{wrapper}
Complementing our positive results, we strengthen the hardness result of \cite{papadimitriou2021online}, and prove $\mathsf{PSPACE}$-hardness of $\alpha$-approximation for some $\alpha<1$ for the \emph{unweighted} Bernoulli~problem. 

Finally, we note that none of the previous philosopher inequalities for online bipartite matching were known to be achievable via pricing-based algorithms. This left a gap in our understanding of polytime approximability of the optimal online mechanism's social welfare paralleling known prophet inequalities' impact on approximation of the optimal \emph{offline} mechanism for such markets. 
Building on the exciting recent work of \cite{banihashem2024power}, we show in \Cref{app-pricing} that our algorithms imply (polytime) truthful mechanisms that provide the same approximation of the optimal online allocation for such bipartite matching markets.

\subsection{Further Related Work}

Prophet inequalities, which compare online algorithms to the benchmark of optimum offline, dominate the literature on Bayesian online resource allocation problems. Since the work of \cite{krengel1978semiamarts} and Samuel-Cahn \cite{samuel1984comparison} for single-item selection, applications in mechanism design and resource allocation have motivated the study of feasibility constraints including combinatorial auctions with XOS/subadditive valuations \cite{feldman2015combinatorial, correa2023subadditive}, (poly-)matroids \cite{duetting2015polymatroid, kleinberg2019matroid}, knapsacks \cite{dutting2020prophet, jiang2022prophet}, general downwards-closed \cite{rubinstein2016prophet} and beyond.

A new research direction considers the benchmark of the optimum \emph{online} algorithm. Perhaps the most fundamental question raised is how polynomial-time algorithms can compare --- the \emph{philosopher inequality}. After the first hardness result for approximating the optimum online benchmark was proved in \cite{papadimitriou2021online} for edge-weighted matching, there has been a line of results on this problem \cite{papadimitriou2021online, saberi2021greedy, braverman2022max, naor2023online}. Recent work also considers the problems of (laminar) matroid Bayesian selection \cite{anari2019nearly, dehaan2024matroid}, and single-item prophet secretary \cite{dutting2023prophet}, and online capacitated matching \cite{braun2024approximating} from this perspective. The optimum online benchmark also is of interest when studying algorithms that are constrained beyond just being running time. For example, recent work on the single-item problem consider algorithms that are threshold based \cite{niazadeh2018prophet} 
or order-unaware \cite{ezra2023next,ezra2023importance,chen2024setting}. 

Although the optimum online is a quite natural benchmark, it remains understudied. This is largely due to the technical challenge it presents compared to the optimum offline (see discussion in \cite[24.4.1]{roughgarden2020beyond}).  
Online matching is a natural case study to develop the new techniques needed.

%% file: prelims.tex
\section{Preliminaries and Technical Overview}\label{sec:prelims}

\paragraph{Problem definition.}
Our input is a random weighted bipartite graph $G=(L,R,w)$. Initially, we are given the set of $n$ nodes in $L=\OFF$, referred to as \emph{offline nodes}. At each time $t\in [T]$, an \emph{online node} $t\in R$ reveals its weights $\vec w^t \in \mathbb{R}_{\ge 0}^{L}$ to the offline nodes, drawn from an \emph{a priori} known distribution, $\vec w^t \sim \calD_t$, and we may match it irrevocably to at most one neighbor $i$, who in turn may be matched to at most one online node. Matching $t$ to offline node $i$ accrues value $w^t_i$.
In the vertex-weighted setting, offline node $i$ has a known weight $w_i \ge 0$ given upfront, and $w^t_i\in \{0,w_i\}$ for each $t$, and $\calD_t$ determines the \emph{neighborhood} of $t$ (i.e., the offline nodes $t$ may be matched to).

For notational simplicity,  we study the Bernoulli case, where each online node $t$ arrives independently with probability $p_t$, in which case it reveals a known set of weights $(w_{i,t})_{i \in L}$, and otherwise its neighborhood is empty (i.e., it reveals a weight of 0 to every offline node). 
This captures the hard examples in \cite{papadimitriou2021online} and our work (in \Cref{sec:hardness}). Moreover, extending results to the non-Bernoulli case is generally syntactic \cite{papadimitriou2021online,braverman2022max,naor2023online}. 
Nonetheless, for completeness, we outline a generalization of \Cref{thm-intro:edge-weighted} to the non-Bernoulli problem in \Cref{app:general}.

\vspace{-0.25cm}
\paragraph{Approximating $OPT_{on}$: An LP Relaxation.}
The optimal online algorithm's value, denoted by $OPT_{on}$ is the solution of a (massive) Markov Decision Process (MDP). The following poly-size (and hence poly-time solvable) linear program \cite{torrico2022dynamic, papadimitriou2021online} was used by all prior polytime algorithms approximating $OPT_{on}$ for our problem \cite{papadimitriou2021online,braverman2022max,naor2023online}. (We provide additional facts regarding \eqref{LP-PPSW}, of possible interest for future work, in \Cref{app:prelims}.)
\begin{align}
    \nonumber \max & \sum_{(i,t)\in E} w_{i,t} \cdot x_{i,t} \tag{LP-OPTon} \label{LP-PPSW} \\
    \textrm{s.t.} 
    &\sum_{i} x_{i, t} \le p_t && \text{for all } t \label{eqn:OnlineMatchingConstraint}\\
    & 0\leq x_{i, t} \le p_t \cdot \left( 1 - \sum_{t' < t} x_{i, t'} \right) && \text{for all } i,t. \label{eqn:PPSWConstraint}
   \end{align}

\begin{prop}\label{lem:LP}
\textup{OPT}\eqref{LP-PPSW}$\geq OPT_{on}$.
\end{prop}
\begin{proof}
Let $x_{i,t}\geq 0$ be the probability that $(i,t)$ is matched by the optimal online algorithm $\calA$. 
Clearly $\vec{x}$ satisfies Constraint \eqref{eqn:OnlineMatchingConstraint}, as $t$ can only be matched if it arrives. To see why $\vec{x}$ 
satisfies Constraint \eqref{eqn:PPSWConstraint}, we first note that 
 $\calA$ can only match $(i,t)$ if $t$ arrives and $i$ was not matched before (with probability $1-\sum_{t'<t}x_{i,t'}$, by disjointness of matching events of $i$). The constraint then follows from independence of these events, since the online algorithm $\calA$ cannot make choices before time $t$ that depend on the arrival of $t$. So, by linearity, the value obtained by the optimal online algorithm, $OPT_{on}$, is the objective of a feasible solution to \eqref{LP-PPSW}, implying the proposition.
\end{proof}

\subsection{Technical overview}\label{sec:technicaloverview}

\eqref{LP-PPSW} precisely captures the optimal online algorithm for instances with a single offline node $i$. 
Specifically, Constraint \eqref{eqn:PPSWConstraint} is sufficient to match each edge with probability precisely $x_{i,t}$, 
by matching with probability $r_{i,t}:=x_{i,t} \cdot (p_t(1-\sum_{t'<t}x_{i,t'}))^{-1}$ if $i$ is still unmatched and $t$ arrives.
The approach of the $(1-1/e)$-approximate algorithm of \cite{braverman2022max} can be seen as performing the above in some sense independently for each offline node $i$,
to obtain ``proposals'' from offline nodes, and then matching $t$ to the highest-valued neighbor; in \Cref{claim:bdmlsimplification}, we give a simpler analysis showing the $1-1/e$ bound. It is not hard to construct tight examples for this approach, even in unweighted instances; for example, for a single online node $t$ with $p_t=1$ and $n$ neighbors and \eqref{LP-PPSW} solution $x_{i,t}=1/n$.
The improvement on this algorithm by \cite{naor2023online} is obtained by maintaining the same marginals while \emph{correlating} the proposals (and careful scaling, discussed below), thus guaranteeing fewer collisions. 
Here, we improve the best known bounds for our problem, by correlating these proposals \emph{optimally}, via a classic offline algorithm: (linear-order) \emph{pivotal sampling} \cite{srinivasan2001distributions,gandhi2006dependent}, whose properties we now state.

\begin{prop}\label{lem:SR}
There exists a polynomial-time algorithm that on set $S=\{s_1,\dots,s_n\}$ and vector $\vec{v} \in [0,1]^{n}$, outputs a random subset $\textup{\textsf{PS}}(S,\vec{v}) \subseteq S$ with $X_i:=\mathds{1}[s_i \in \textup{\textsf{PS}}(S,\vec{v})]$ satisfying:
\begin{enumerate}[label=(P{{\arabic*}})]
    \item \textbf{Marginals:} $\Pr[X_i=1] = v_i$ for all $s_i\in S$.\label{level-set:marginals}
    \item \textbf{Prefix property:} $\Pr[ |\textup{\textsf{PS}}(S,\vec{v}) \cap \{s_1,\dots,s_k\}|\geq 1] = \min(1,\,\sum_{i\leq k} v_i)$ for all $k \leq n$. \label{level-set:prefix}
    \item \label{level-set:neg-corr} \textbf{Negative cylinder dependence (NCD):} For all $I\subseteq S$,
    $$\Pr \left[ \bigwedge_{i\in I} (X_i=1)\right]\leq \prod_{i\in I} \Pr[X_i=1] \qquad \textrm{and} \qquad \Pr\left[\bigwedge_{i\in I} (X_i= 0)\right]\leq \prod_{i\in I} \Pr[X_i=0].$$
\end{enumerate}
\end{prop}

Pivotal sampling
repeatedly picks the two lowest indices $i\neq j$ with fractional $v_i,v_j$ and randomly sets one of these to 0 or 1 while preserving their sum and expectations.
Our algorithm uses pivotal sampling both explicitly in its description and implicitly in its analysis.
Algorithmically, at each time $t$ we apply pivotal sampling 
to the vector $(r_{i,t})_{i \text{ available}}$, sorted by decreasing $w_{i,t}$, so as to:
\begin{enumerate}[label=(Q{{\arabic*}})]
\item Assign/discard each available offline node to $t$ (if it arrives) with the correct marginal $r_{i,t}$. 
\label{prop-intro:marginals} 
\item \label{prop-intro:min(1,R)}
Get $\Pr[\textrm{match $t$ to an $i$ with $w_{i,t}\geq w$}] = p_t\cdot \min(1,R_{t,w})$, for any arrival $t$, and $w \ge 0$, where $R_{t,w} := \sum_{\stackrel{i \text{ available}}{w_{i,t} \ge w}} r_{i,t}$. 
\item Obtain strong negative correlations between offline nodes' availability, specifically NCD. \label{prop-intro:ncd}
\end{enumerate}

\noindent\textbf{Optimality.} By Property \ref{prop-intro:marginals},
the bound of Property \ref{prop-intro:min(1,R)} is optimal,
as an arriving $t$ cannot yield value $\geq w$ with probability greater than $\min(1,R_{t,w})$, by the union bound.

\subsubsection{Overview of  the analysis}
By Properties \ref{prop-intro:marginals}-\ref{prop-intro:ncd}, the core analytic challenge for our philosopher inequality is to provide tail expectation bounds for $X$ the sum of negatively correlated $[0,1]$-weighted Bernoulli~variables, specifically, lower bounding $\E[\min(1,X)]=\E[X]-\E[(X-1)\cdot \mathds{1}[X>1]]$.
We therefore provide a number of such tail bounds, of possible independent interest. To emphasize the generality of these probabilistic inequalities, we note that when $\E[X]=1$ (the worst case for our analysis), this tail expectation is half of the \emph{mean absolute deviation}, $\E[|X-\E[X]|]$. The latter is a notion of dispersion studied intently by statisticians (see \cite{berend2013sharp} and references therein), and often used by theoretical computer scientists in varied contexts (see, e.g., \cite{ambainis2014quantum,gupta2019stochastic,braverman2020coin}).

In what follows, we let $X = \sum_i c_i \cdot X_i$, where $\{X_i\}$ are NCD Bernoullis, and $c_i \in [0,1]$ for all~$i$.

\paragraph{Independent Coin Bound.} 
Let $S$ denote the random set of all $i$ such that $X_i = 1$; by the union bound, for any realization of $S$, we have that $\min \left( 1, \sum_{i \in S} c_i \right) \ge 1 - \prod_{i \in S} (1-c_i).$ 
Thus, by imagining that every $i\in S$ independently flips a $\Ber(c_i)$ variable, 
we obtain the following bound on $\E[\min(1,X)]$, which we fittingly call the \emph{independent coin bound}.
\begin{align}\label{eqn:intro:independent-proposals}
    \E[\min(1,X)] \ge \E_S \left[ 1 - \prod_{i \in S} (1 - c_i) \right] = 1 - \E \left[ \prod_i (1-c_i \cdot X_i) \right]  \ge 1 - \prod_i (1 - c_i\cdot \E[X_i]).
\end{align}
Here, the final step follows from non-trivial calculations which crucially use that $\{X_i\}$~are~NCD. Using that $1 - c_i\cdot \E[X_i]\leq \exp(-\E[c_i\cdot X_i])$, together with convexity, the RHS above (in essence the same as in \cite{braverman2022max}) suffices to obtain a $1-1/e$ approximation, but no better (see \Cref{claim:indprop1-1/e}).

\paragraph{Bucketing Bound.} The first inequality of the independent coin bound of \Cref{eqn:intro:independent-proposals} might be quite loose, e.g., if all $c_i$'s are small. We argue that this bound can be tightened if the $c_i$'s can be non-trivially partitioned into a small number of buckets $\mathcal{B}$ such that for any individual $B \in \mathcal{B}$ we have $\sum_{i \in B} c_i \le 1$. In the \emph{bucketing bound} (\Cref{lem:bucketed-independent-proposal-bound}), we show that in fact 
\begin{align}\label{eqn:intro:bucketing-bound}
\E[\min(1,X)] \ge 1 - \prod_{B \in \mathcal{B}} \left( 1 - \sum_{i \in B} c_i \cdot \E[X_i] \right).
\end{align}
This bound is similar to the bound obtained by \cite{naor2023online} for their algorithm, though their bound holds for a \emph{single}, \emph{explicit} bucketing $\mathcal{B}$ chosen in advance for each online arrival, and used to inform the design of the algorithm.
Our bound holds \emph{implicitly} for \emph{all} valid bucketings simultaneously. 

\paragraph{Fractional Bucketing Bound.} 
The bucketing bound of \Cref{eqn:intro:bucketing-bound} unfortunately can also be quite lossy, due to integrality issues. For example, if $c_i = 0.51$ for every $i$ then no non-trivial bucketing occurs, and we revert to the independent bound.
Our most novel bound on $\E[\min(1,X)]$ allows us in some sense to pack fractions of these $c_i$ together.
In particular, for any time-step $t$ and threshold $\theta\in [0,1]$, for $C:=\sum_{i: q_i\geq 1-\theta} c_i\cdot q_i$ the expected sum of weights of high-probability variables, then,  for $\{z\}=z-\lfloor z\rfloor$ the fractional part of $z$, we prove the following bound.

\begin{align}\label{eqn:intro:fractional-bucketing-bound}
    \E[\min(1,X)]\geq 1-\left(1-(1-\theta)\cdot \left\{\frac{C}{1-\theta}\right\}\right)\cdot \theta^{\lfloor \frac{C}{1-\theta}\rfloor}\cdot \prod_{i: q_i < 1-\theta} (1-c_i\cdot q_i).
\end{align}
Interestingly, the key step in the transformation is again (a generalization of) pivotal sampling, but this time only used implicitly to obtain this simpler instance. 
In particular, using this generalized pivotal sampling, we can modify the $\{c_i \mid q_i\geq 1-\theta\}$ so that at most one of these is not binary, while preserving expectations of these (now random) $c_i$ (and thus $C$), and without increasing the expression $\E[\min(1,X)]$ (and indeed any concave expression in $X$). This then results in a surrogate instance where the independent coin bound yields \Cref{eqn:intro:fractional-bucketing-bound} for the initial instance. 

\paragraph{Our algorithmic results.} 
We now circle back to our original problem. 
Our algorithm as outlined only yields a $(1-1/e)$-approximation for edge-weighted matching, by a similar example to the tight example of \cite{braverman2022max} (see \Cref{example:need-to-rescale}). 
To overcome this, we combine our fractional bucketing bound of \Cref{eqn:intro:fractional-bucketing-bound} 
with the scaling approach of \cite{naor2023online}:
decreasing values $x_{i,t}$
for times $t$ when $i$ has low fractional degree $\sum_{t'<t} x_{i,t'}$ at most $\theta$, and increasing values $x_{i,t}$ later. 
This way, every $t$ either has much $x$-flow from low-degree neighbors, in which case the fractional bucketing part of \Cref{eqn:intro:fractional-bucketing-bound} yields a significant boost, or much flow comes from high-degree neighbors, whose value is boosted and similarly provides a boost.

Our vertex-weighted algorithm avoids the above scaling step. 
Its analysis relies on a local-to-global convex averaging argument, showing that a high  lower bound on $\E[\min(1,R_t)]$ on average suffices for a similar approximation ratio (\Cref{lem:convexavgacrosst}).
We combine this averaging argument with the observation that on average, either much of the \eqref{LP-PPSW} solution induces sequences of low-variance $X$, which we show in \Cref{lem:variance-bound} implies a high lower bound on $\E[\min(1,X)]/\E[X]$, or the $c_i\cdot q_i$ terms are high on average, in which case a consequence of H\"older's inequality combined with our fractional bucketing bound yields a high approximation on average.

%% file: algo.tex
\section{The Algorithm}\label{sec:algo}

In this section we introduce our core algorithm and prove some key properties needed for its analysis in later sections.

Our algorithm takes as input a solution $\vec x$ satisfying \eqref{LP-PPSW}~Constraint~\eqref{eqn:PPSWConstraint} (though possibly not the others).
The algorithm maintains a matching $\calM$ and set $F_t$ of \emph{free} offline nodes before time $t$. Let $F_{i,t}:=\mathds{1}[i\in F_t]$ denote the matched status of offline node $i$ before time $t$. 
The algorithm strives to remove offline nodes from the set of free nodes at time $t$ by matching them or \emph{discarding} them, so as to guarantee (i) a closed form for the probability of $i$ to be free, $\Pr[F_{i,t}]=1-y_{i,t}$, where $y_{i,t} := \sum_{t' < t} x_{i,t'}$, (ii) negative correlations between the offline nodes' free statuses, and (iii) maximum expected weight of each match, conditioned on the above.
To this end, at each time $t$, we apply pivotal sampling to the vector $(r_{i,t}:=x_{i,t}/(p_t\cdot (1-y_{i,t})))_i$ indexed by free $i\in F_t$ sorted in decreasing order of $w_{i,t}$, to pick a set of offline \emph{proposers} to $t$, denoted by $I_t$. We then match $t$ to a highest-weight such proposer $i^*_t\in I_t$ if $t$ arrives, and, independently of $t$'s arrival, discard all other proposers $i\in I_t\setminus\{i^*_t\}$ with probability $p_t$.
The algorithm's pseudocode is given in \Cref{alg:proposals-core}.

\begin{algorithm}[H]
	\caption{Online Correlated Proposals}
	\label{alg:proposals-core}
	\begin{algorithmic}[1]
 \Statex \textbf{Input:} A vector $\vec x$ satisfying \eqref{LP-PPSW}~Constraint~\eqref{eqn:PPSWConstraint} 
 \medskip
\State $\calM\gets \emptyset,\;  F_1\gets [n]$
		\Comment{$\calM$ is the output matching}
		\ForAll{times $t$} \label{line:loop-start}
        \State $F_{t+1}\gets F_t$ \Comment{initially, no free node $i\in F_t$ is matched/discarded before time $t+1$}
		\ForAll{offline nodes $i$}
		\State $r_{i,t} \gets \frac{x_{i,t}}{p_t\cdot \left(1-\sum_{t'<t} x_{i,t'}\right)}$
        \EndFor
        \State Let $\vec v$ be the vector $(r_{i,t})_{i \in F_t}$ indexed by $i$ sorted in decreasing order of $w_{i,t}$
        \State Let $I_t \gets \textsf{PS}(F_t, \vec v)$ \label{line:calltoSR}
        \If{$I_t\neq \emptyset$}
        \State Pick some $i_t^*\in \arg\max_{i\in I_t} \{w_{i,t}\}$
        \If{$t$ arrives}
        \State Add $(i^*_t,t)$ to the matching $\calM$ and set $F_{t+1}\gets F_{t+1}\setminus \{i^*_t\}$ \Comment{match $i^*_t$}
        \EndIf
\label{line:match-end}
        \ForAll{$i \in I_t \setminus \{i^*_t\}$}
        \State  With probability $p_t$ (independently) set $F_{t+1}\gets F_{t+1}\setminus \{i\}$  \Comment{discard $i$} \label{line:loop-end}
        \EndFor
        \EndIf
        \EndFor
	\end{algorithmic}
\end{algorithm}	

We first note that \Cref{alg:proposals-core} is well-defined, by Constraint~\eqref{eqn:PPSWConstraint} guaranteeing that for every $i \in F_t$ we have that $r_{i,t} \in [0,1]$. This is necessary for invoking pivotal sampling (\Cref{lem:SR}). We can also easily observe that the algorithm matches/discards offline nodes precisely according to the input vector $\vec x$ (see \Cref{app:algo} for the short inductive proof). 

\begin{restatable}{lem}{obsmatchmarginals} \label{obs:matchmarginals}\label{cor:Fit}
For every pair $(i,t)$, we have that $\Pr[F_{i,t}]=1-y_{i,t}$. 
\end{restatable}

To gain intuition for the effect of discarding in  \Cref{alg:proposals-core}, we quickly show that if we replace our call to pivotal sampling in \Cref{line:calltoSR} by \emph{independent sampling} (i.e., each $i$ is sampled with probability $r_{it}$ independently), we maintain full independence between offline nodes. This in turn lets us argue in one paragraph that the algorithm is $(1-1/e)$-approximate, significantly simplifying the result of \cite{braverman2022max}. 

\begin{claim}
\Cref{alg:proposals-core} run on an optimal solution to \textup{\eqref{LP-PPSW}}, but with the call to $\textup{\textsf{PS}}(\cdot, \cdot)$ in \Cref{line:calltoSR} replaced by independent sampling, yields a $(1-1/e)$-approximate philosopher inequality. 
 \label{claim:bdmlsimplification}
\end{claim}
\begin{proof}
    With independent sampling, we claim that for every $t$ the indicators $\{F_{i,t}\}_i$ are independent. This can be observed by a simple coupling argument: if we ignore the arrival status of $t$, and simply discard every proposing offline node with an independent $\text{Ber}(p_t)$, independence is trivial. Consider the coupling that replaces for every $t$, one of the independent $\text{Ber}(p_t)$ variables with the arrival status of $t$, for the top weight proposing $i$. Observe $t$'s arrival status is independent of all previous proposals, free statuses, and remaining $\text{Ber}(p_t)$'s.  

    Because we sample with the same marginals as pivotal sampling, for every $(i,t)$ we have $\Pr[F_{i,t}] = 1-y_{i,t}$ as in \Cref{cor:Fit}. If $W_t$ denotes the (random) weight $t$ of $t$'s match, then we note the probability $W_t$ is at least some fixed $w$ is 
    \begin{align*}
       p_t \cdot \left( 1 - \prod_{i : w_{i,t} \geq w} \left( 1 - \Pr[F_{i,t}] \cdot r_{i,t} \right) \right)  &= p_t \cdot \left( 1 - \prod_{i : w_{i,t} \geq w} \left( 1 - \frac{x_{i,t}}{p_t} \right) \right) \\
       &\ge p_t \cdot \left( 1 - \exp \left( - \sum_{i: w_{i,t} \ge w} \frac{x_{i,t}}{p_t} \right) \right), 
    \end{align*}
    which in turn is at least $(1-1/e) \cdot \sum_{i: w_{i,t} \ge w} x_{i,t}$ by convexity and Constraint~\eqref{eqn:OnlineMatchingConstraint}. Thus $$\E[w(\mathcal{M}_t)] = \int_{0}^\infty \Pr[W_t \ge w] \, dw \ge \int_0^{\infty} (1-1/e) \cdot \sum_i x_{i,t} \cdot \mathbbm{1}[w_{i,t} \ge w] \, dw = (1-1/e) \cdot \sum_i x_{i,t} w_{i,t}$$ which demonstrates a $(1-1/e)$-approximation to \text{OPT}\eqref{LP-PPSW}, and hence $\opton$. 
\end{proof}
	
	In this work, we prove that better guarantees are possible with pivotal sampling. For any $w\geq 0$, we let $R_{t,w} := \sum_{i: w_{i,t}\geq w} r_{i,t}\cdot F_{i,t}$ be the ``request'' at time $t$ of offline nodes $i$ with edge $(i,t)$ of weight $w_{i,t} \geq w$.
 With this notation, we have the following. 
	\begin{lem}\label{per-t-min-bound}
		For any time $t$ and value $w\geq 0$, 
		$$\Pr[\calM\ni (i,t):\;  w_{i,t}\geq w] = p_t\cdot \E[\min(1,R_{t,w})].$$
	\end{lem}
	\begin{proof}
        Follows directly from Property \ref{level-set:prefix} of pivotal sampling (\Cref{lem:SR}), together with $t$ being matched to a highest-weight offline neighbor $i\in I_t$ if it arrives.
    \end{proof}

This motivates the study of sums $R_{i,t}:=r_{i,t}\cdot F_{i,t}$ for our analysis, and in particular tail expectation bounds. Crucial to our bounding is the lemma that the variables $\{F_{i,t}\}_i$ are negative cylinder dependent (NCD).

\begin{restatable}{lemma}{ncd}\label{lem:NCD}
For any time $t$, the variables $\{F_{i,t}\}_{i}$ are NCD. That is, for any $I \subseteq \OFF$, $$\Pr \left[ \bigwedge_{i \in I} F_{i,t} \right] \le \prod_{i \in I} \Pr[F_{i,t}] = \prod_{i \in I} (1-y_{i,t}) \quad \quad \text{ and } \quad \quad \Pr \left[ \bigwedge_{i \in I} \overline{F_{i,t}} \right] \le \prod_{i \in I} \Pr[\overline{F_{i,t}}] = \prod_{i \in I} y_{i,t}.$$ 
\end{restatable}

The proof of \Cref{lem:NCD} uses an inductive argument; at a high level it bears resemblance to bounds such as those in \cite{cohen2018randomized, braverman2022max}, although it differs from prior work in the reliance on properties of pivotal sampling and the independent discarding. We defer a detailed proof to \Cref{app:algo}, and provide here a brief intuition behind (half of) the proof upper bounding the probability all $i\in I$ are not free at time $t+1$. The critical claim is that the probability the nodes in the free subset $S = I\cap F_t$ are all matched/discarded at time $t$ is at most $\prod_{i \in S} (p_t\cdot r_{i,t}).$ To show this, we first note that for any history resulting in $S$ all being free at time $t$, 
the probability all $i\in S$ propose is at most $\prod_{i\in S} r_{i,t}$, by Properties \ref{level-set:marginals} and \ref{level-set:neg-corr} of pivotal sampling. 
The crucial claim then follows from the above and our \emph{independent discarding}, which has every proposing node $i\in I_t$ be matched or discarded at time $t$ with an independent $\text{Ber}(p_t)$ coin flip, regardless of whether it is the top-weight proposer (saved for the possible arrival at $t$) or a lower-weight proposer.

%% file: Emin1X.tex
\newpage \section{Tail Expectation Bounds: Lower Bounding \texorpdfstring{$\E[\min(1,X)]$}{E[min(1,X)]}}\label{sec:Emin1X}

In this section, motivated by Lemmas \ref{per-t-min-bound} and \ref{lem:NCD}, we prove lower bounds on $\E[\min(1,X)]$, thus upper bounding the tail expectation $\E[(X-1)\cdot \mathds{1}[X>1]]$, for the following general setup:\footnote{In our algorithm's analysis we are concerned with $X=R_{t,w}$ and $X_i=F_{i,t}$. This more general notation simplifies the discussion later and highlights the generality of the tail bounds studied in this section.}

\begin{Def}\label{sec-setup}
In this section, $X:=\sum_{i=1}^n c_i \cdot X_i$, where $c_i\in [0,1]$ and $\{X_i\sim \textup{Ber}(q_i)\}$ are NCD.
\end{Def}

\subsection{The independent coin and bucketing bounds}

Our first bound on $\E[\min(1,X)]$ is as one would obtain for independent $\{X_i\cdot \Ber(c_i)\}$. 
We therefore refer to this as the independent coin bound. 

\begin{restatable}{lem}{indproposals}
\label{lem:basic-independent-proposal-bound}
Let $X,X_1,\dots,X_n$ be as in \Cref{sec-setup}. Then,
$$\E[\min(1,X)] \geq 1-\prod_{i\in \OFF} \left(1-c_i\cdot q_i\right).$$\end{restatable}

The above already implies a bound of $1-1/e$ for our algorithmic problem, but no better (see \Cref{claim:indprop1-1/e}). 
The following \emph{bucketing bound} generalizes the above.
\begin{restatable}{lem}{indproposalsbucketed}
\label{lem:bucketed-independent-proposal-bound}
Let $X,X_1,\dots,X_n$ be as in \Cref{sec-setup}, and let $\calB$ be a partition of $\OFF$ such that $\sum_{i \in B} c_i \le 1$ for every bucket $B \in \calB$. Then, $$\E[\min(1,X)] \geq 1-\prod_{B \in \calB} \left(1- \sum_{i \in B} c_i\cdot q_i\right).$$
\end{restatable}

We provide a direct proof of \Cref{lem:basic-independent-proposal-bound} here, and defer a slightly more notation-heavy generalization of this proof, yielding \Cref{lem:bucketed-independent-proposal-bound}, to \Cref{app:Emin1X}. The proof of \Cref{lem:basic-independent-proposal-bound} uses the following two facts about real numbers, both provable via probabilistic arguments. 

\begin{restatable}[Union Bound]{fact}{unionbound} \label{fact:unionbound}
If $0 \le c_1, c_2, \ldots c_n \le 1$, then  $\min \left( 1, \,\sum_{i=1}^n c_i \right) \ge 1 - \prod_{i=1}^n \left( 1 -  c_i \right).$
\end{restatable}

\begin{restatable}{fact}{probabilisticfact2}\label{fact:probabilisticfact2}
Let $S$ be a random subset of $[n]$ and $0 \le c_1, c_2, \ldots, c_n \le 1$. Then, $$\sum_{J \subseteq [n]} \Pr[S = J] \cdot \prod_{i \in J} \ \left( 1 -  c_i \right) = \sum_{J \subseteq [n]} \Pr[S \subseteq J] \cdot \prod_{i\not\in J} c_i\cdot \prod_{i \in J} \left(1-  c_i\right).$$
\end{restatable}
\begin{proof}
For each $i$ independently flip a biased coin that is heads with probability $c_i$. The LHS and RHS both count the probability that every element in $S$ flipped tails. 
\end{proof}

\begin{proof}[Proof of \Cref{lem:basic-independent-proposal-bound}]
Let $S := \{i \mid X_i = 1\}$. Using the above facts, we get our desired~tail~bound.
\begin{align}
   \E[\min(1,X)]  &= \sum_{J \subseteq [n]} \Pr[S = J] \cdot \min \left(1,\, \sum_{i \in J} c_i \right) \nonumber \\
    &\ge \sum_{J \subseteq \OFF} \Pr[S = J] \cdot \left( 1 - \prod_{i \in J} (1 - c_i)  \right) && \text{\Cref{fact:unionbound}} \nonumber \\
    &= 1 - \sum_{J \subseteq \OFF} \Pr[S = J] \cdot  \prod_{i \in J} (1 -c_i) \nonumber \\
    &=  1 - \sum_{J \subseteq \OFF} \left( \Pr[ S \subseteq J] \cdot \prod_{i \not\in J} c_i \cdot  \prod_{i \in J} (1 - c_i) \right)  && \text{\Cref{fact:probabilisticfact2}} \nonumber \\
    &\ge  1 - \sum_{J \subseteq \OFF} \left( \prod_{i \not\in J} (1-q_i) \cdot \prod_{i \not\in J} c_i \cdot  \prod_{i \in J} (1 -c_i) \right) && \{X_i\} \text{ NCD} \label{eqn:crucial-NCD} \\
    &=  1 - \prod_{i \in \OFF} \left( (1-q_i) \cdot c_i + 1 - c_i\right) && \textrm{multi-binomial theorem} \nonumber \\
    & = 1 - \prod_{i \in [n]} \left( 1 - c_i \cdot q_i \right). &&  \qedhere \nonumber
\end{align}
\end{proof}

\subsection{``Fractional'' bucketing bound via generalized pivotal sampling}

As noted in \Cref{sec:prelims}, the bucketing bound can be lossy due to integrality issues. The following lemma allows us to transform (in our analysis only) the weighted Bernoullis $c_i\cdot X_i$, resulting in a simpler instance implying a stronger tail bound for the original instance.

\begin{lem}\label{lem:pipage}
Let $X = \sum_{i=1}^n c_i \cdot X_i$, where $X_1, \ldots, X_n$ are (possibly correlated) Bernoulli random variables with $c_i \in [0,1]$ for all $i$, with $0 < c_j, c_k < 1$ for indices $j\neq k$. Let $\rho^+ \ge 0$ and $\rho^- \ge 0$ be the largest real numbers such that $c_j + \rho^+ \cdot \E[X_k]$, $c_k - \rho^+ \cdot \E[X_j]$, $c_j - \rho^- \cdot \E[X_k]$ and $c_k + \rho^- \cdot \E[X_j]$ are all in $[0,1]$. Define the following two variables:
\begin{align*}X^+ & := \left(c_j + \rho^+ \cdot \E[X_k] \right) \cdot X_j + \left(c_k - \rho^+ \cdot \E[X_j] \right) \cdot X_k + \sum_{i \in [n] \setminus \{j,k\}} c_i \cdot X_i,\\
X^- & := \left(c_j - \rho^-\cdot \E[X_k] \right) \cdot X_j + \left( c_k + \rho^- \cdot\E[X_j] \right) \cdot X_k + \sum_{i \in [n] \setminus \{j,k\}} c_i \cdot X_i.
\end{align*}
Then, $\E[X]=\E[X^+]=\E[X^-]$, and for any concave function  $f(\cdot)$, $$\E[f(X)] \ge \min(\E[f(X^+)],\E[f(X^-)]).$$  
\end{lem}

\begin{proof}
Note that $X = \sigma \cdot X^+ + (1 - \sigma) \cdot X^-$ for $\sigma = \frac{\rho^-}{\rho^+ + \rho^-}$. Therefore,  by concavity of $f(\cdot)$,  
\begin{align*}
\E[f(X)] & \ge \sigma \cdot \E[f(X^+)] + (1-\sigma) \cdot \E[f(X^-)]\geq \min(\E[f(X^+)],\E[f(X^-)]).\qedhere 
\end{align*}
\end{proof}

Repeatedly applying the preceding lemma to the weights $c_j,c_k$ of pairs of variables $X_j,X_k$ with $\E[X_j]=\E[X_k]\geq 1-\theta$, we can obtain a new weight vector $c'$ with at most one fractional value, with the new variable $X'=\sum_i c'_i\cdot X_i$ satisfying $\E[\min(1,X)]\geq \E[\min(1,X')]$. 
The following lemma asserts that the worst-case scenario for this bound is when $\E[X_i]=1-\theta$ for all $X_i$ with $\E[X_i]\geq 1-\theta$ and $c'_i\neq 0$. 
Recall that~$\{z\}=z - \lfloor z \rfloor$ denotes the fractional part of $z$. 

\begin{restatable}{lem}{amgmforbucketing}\label{claim:amgmbasedineq}
    Let  $\theta, c \in [0,1]$, and $\{ q_i\}_{i \in A}$ be real numbers such that $q_i \in [1 - \theta, 1]$ for all $i$. Finally, let $\mu(c,\vec{q}) := c\cdot(1 - \theta) + \sum_{i \in A} q_i.$ Then, for $\mu:=\mu(c,\vec{q})$, we have
    $$(1 - c\cdot(1 - \theta)) \cdot \prod_{i \in A} (1 - q_i) \le \left( 1 - (1-\theta) \cdot \Big\{ \frac{\mu}{1-\theta} \Big\} \right) \cdot \theta^{\lfloor \frac{\mu}{1-\theta} \rfloor}.$$
\end{restatable}

To gain intuition about this lemma, note that it is tight when $q_i = 1-\theta$ for all $i$; in \Cref{app:Emin1X} we show that this is indeed the worst case, using a local exchange argument. 

\newpage 
Armed with the above, we are now ready to prove our main lower bound on $\E[\min(1,X)]$. 

\begin{restatable}{lem}{fracbucketing}\label{lem:fracbucketingbound}
Let $X,X_1,\dots,X_n$ be as in \Cref{sec-setup}. For threshold $\theta \in [0,1]$, we consider the set $S :=\{i\in \OFF \mid q_i \ge 1-\theta\}$ and $\mu_S := \sum_{i \in S} c_i \cdot q_i$. Then, $$\E[\min(1,X)] \ge 1 - \left( 1 - (1-\theta) \cdot \Big\{ \frac{\mu_S}{1-\theta} \Big\} \right) \cdot \theta^{\lfloor \frac{\mu_S}{1-\theta} \rfloor} \cdot \prod_{i \notin S} \left( 1 - c_i \cdot q_i \right) .$$
\end{restatable}

\begin{proof}
    By \Cref{lem:pipage}, repeatedly applying the generalized pivotal sampling step of increasing and decreasing $c_j$ and $c_k$ yields a new set of weights $\vec{c}'\in \mathbb{R}^n$ with $X':=\sum_i c'_i\cdot X_i$ satisfying $\E[\min(1,X)]\geq \E[\min(1,X')]$. 
    Moreover, we also have  (deterministically) that $\sum_{i\in S} c'_i\cdot X_i=\mu_S$ and that $c'_i\not\in \{0,1\}$ for at most one $i^*\in S$. 

Now, using the independent coins bound of \Cref{lem:basic-independent-proposal-bound} (which coincidentally is tight if all but one coefficient is binary), we find that indeed 
    \begin{align*}
        \E[\min(1,X)] &\ge \E[\min(1,X')] && \textrm{\Cref{lem:pipage}}\\
        & \geq \E\left[1 - \prod_{i} (1 - c'_i \cdot q_i)\right] && \textrm{\Cref{lem:basic-independent-proposal-bound}}\\
        &\ge 1 - \left( 1 - (1-\theta) \cdot \Big\{ \frac{\mu_S}{1-\theta} \Big\} \right) \cdot \theta^{\lfloor \frac{\mu_S}{1-\theta} \rfloor} \cdot  \prod_{i \notin S} (1 - c_i \cdot q_i) . && \text{\Cref{claim:amgmbasedineq}} 
    \end{align*}
    To justify our use of \Cref{claim:amgmbasedineq} in the final line, note that while the index $i$ with fractional coefficient $c'_i$ may not have $q_i = (1 - \theta) $ after the transformations of \Cref{lem:pipage}, after using the independent coins bound the resulting expression will be the same with syntactic changes.\footnote{In particular, write $(1 - c'_i \cdot q_i) = (1 - c_i'' \cdot q_i'')$, where $c_i''$ is taken to be as large as possible with the caveats that it should not exceed $1$, and $q_i''$ should not drop below $(1-\theta)$. One of these two inequalities becomes tight first, and both cases fit the form of \Cref{claim:amgmbasedineq} (the first has the fractional coefficient corresponding to probability $(1-\theta)$, and the second has everyone in $S$ with binary coefficients).}
\end{proof}

\subsection{Variance-based bound}
For our final approach to bounding $\E[\min(1,X)]$, we note that assuming $\E[X] \le 1$, Jensen's inequality implies that $\E[\min(1,X)]\geq \E[X]-\frac{1}{2}\cdot \sqrt{\text{Var}(X)}$, 
though this can be loose if $\E[X]<1$.
The following lemma refines this bound (even for $X$ not as in \Cref{sec-setup}). 

\begin{restatable}{lem}{varbound}\label{lem:Emin1X-variance}
Let $X$ be a non-negative random variable with $\E[X]\leq 1$. Then, $$\E[\min(1,X)]\geq \E[X] - \frac{1}{2 } \cdot \sqrt{\Var( X)\cdot \E[X]}.$$
\end{restatable}
The proof of \Cref{lem:Emin1X-variance}, obtained by considering a two-point distribution with the same expectation as $X$ which is supported on only $\E[X\mid X\leq 1]$ and $\E[X\mid X>1]$, is deferred to \Cref{app:Emin1X}.

The bound of \Cref{lem:Emin1X-variance} is generally incomparable with \Cref{lem:fracbucketingbound} (see \Cref{sec:incomparable}). 
In \Cref{sec:vertexweighted} we combine both to analyze our algorithm for the vertex-weighted and unweighted settings. For now, we turn to applying the bound of \Cref{lem:fracbucketingbound} to obtain our edge-weighted result, in the next section.

%% file: edgeweighted.tex
\newpage
\section{The Edge-Weighted Algorithm} \label{sec:edge-weighted}

In this section we provide our application of \Cref{alg:proposals-core} to the general edge-weighted problem.

\subsection{Rescaling the LP solution}\label{sec:scaling}

The following illustrative example of \cite{braverman2022max} shows that running \Cref{alg:proposals-core} on an optimal solution $\vec{x}$ to \eqref{LP-PPSW} could yield only a $(1-1/e)$-approximation. 
\begin{example}\label{example:need-to-rescale}
    Consider an instance with $n$ offline nodes and $n+1$ online nodes; for $t\in [n]$, the $t$\textsuperscript{th} online node neighbors only offline node $i=t$ with $w_{i,t}=1$ and arrives with probability $1-1/n$. The last online node neighbors all offline nodes with large weight $W\gg 1$, and arrives with probability one. The unique optimal LP solution sets $x_{i,i} = 1-1/n$ and $x_{i, n+1} = 1/n$ for all $i \in [n]$. In this case, we have that $ R_{n+1} \sim \Bin(n,1/n)$, and so $\E[\min(1,R_{n+1})] = 1-(1-1/n)^n\to 1-1/e$. 
\end{example}

In \Cref{example:need-to-rescale}, informally, the ``early'' edges of offline nodes $i$ (i.e., edges to nodes $t$ where $y_{i,t}$ is low) result in a match with probability close to $x_{i,t}$, while ``later'' edges have a a much lower matching probability than $x_{i,t}$. 
This motivates the rescaling approach of  \cite{naor2023online}: decreasing $x_{i,t}$ for early edges $(i,t)$ and increasing $x_{i,t}$ for late edges. 
For completeness, we introduce (and slightly generalize) this rescaling approach here, which we adopt shortly.

\begin{Def}\label{def:scaling-gen}
For non-decreasing function $f: [0,1] \mapsto \mathbb{R}_{\ge 0}$ with $\int_0^1 f(z) \, dz = 1$, define for~all~$i,t$:
\begin{align}
\hat{x}_{i,t} & := \int_{y_{i,t}}^{y_{i,t}+x_{i,t}} f(z) \, dz,\label{eqn:scaling} \\
\hat{y}_{i,t} & := \sum_{t' < t} \hat{x}_{i,t'} = \int_{0}^{y_{i,t}} f(z) \, dz.
\end{align} 
\end{Def}

We first show that the transformation $\vec{x}\mapsto \vec{\hat{x}}$ preserves Constraint \eqref{eqn:PPSWConstraint}, i.e.,  $\hat{r}_{i,t}:=\frac{\hat{x}_{i,t}}{p_t(1-\hat{y}_{i,t})}$ is in $[0,1]$ if $r_{i,t}\in [0,1]$. Non-negativity is trivial, while the upper bound is proven in the following. 
\begin{claim} \label{hatxwelldefined}
    If $\vec{x}$ satisfies Constraint \eqref{eqn:PPSWConstraint}, then $\hat{x}_{i,t} \le p_t \cdot \left( 1 - \hat{y}_{i,t} \right)$ for all $i,t$.
\end{claim}

\begin{proof}
Using the definition of $\hat{x}_{i,t}$, monotonicity of $f(\cdot)$, Constraint \eqref{eqn:PPSWConstraint} and $\int_{0}^1 f(z) \; dz = 1$, and finally the definition of $\hat{y}_{i,t}$, we obtain our desired bound. 
\begin{align*}
    \hat{x}_{i,t} &= \int_{y_{i,t}}^{y_{i,t} + x_{i,t}} f(z) \, dz  \le \frac{x_{i,t}}{1-y_{i,t}} \cdot \int_{y_{i,t}}^1 f(z) \, dz 
    \le p_t \cdot \left( 1 - \int_0^{y_{i,t}} f(z) \, dz \right) = p_t \cdot (1 - \hat{y}_{i,t} ).  \qedhere
    \end{align*}
    
\end{proof}

Thus, $\vec{\hat{x}}$ is still a valid input for \Cref{alg:proposals-core} (though it is not a valid solution to \eqref{LP-PPSW}). This motivates our \Cref{alg:proposals-edge-weighted} for the edge-weighted problem, where we rescale according to \Cref{def:scaling-gen}, with $f(\cdot)$ a step function, as in \cite{naor2023online}. 
\begin{Def}\label{def:step}
For some $\eps,\delta\in [0,1]$ to be determined later and $\theta := \frac{\delta}{\delta + \eps}$, we set:
\begin{align*}
f(z) := \begin{cases} 1 - \eps & z \in [0, \theta]\\
1 + \delta & z \in (\theta, 1].
\end{cases}
\end{align*}
Note that by the choice of $\theta$, we have that $\int_0^1 f(z) \, dz = 1$.
\end{Def}

\begin{algorithm}[H]
	\caption{Online Correlated Proposals with Scaling}
	\label{alg:proposals-edge-weighted}
	\begin{algorithmic}[1]
\State Compute optimal solution $\vec{x}$ to \eqref{LP-PPSW}
\State For every $(i,t)$, compute $\hat{x}_{i,t}$ according to Definitions \ref{def:scaling-gen} and \ref{def:step}
\State Run \Cref{alg:proposals-core} using $\vec{\hat{x}}$
\end{algorithmic}
\end{algorithm}

\subsection{Edge-weighted analysis}

In this section we provide our main result for edge-weighted matching.

\begin{restatable}{thm}{thmedgeweighted}\label{thm:edge-weighted}
\Cref{alg:proposals-edge-weighted} with appropriate $\eps,\delta\in [0,1]$ is a polynomial-time $\edgeapprox$-approximate online algorithm for edge-weighted online stochastic bipartite matching.   
\end{restatable}

The above algorithm's approximation ratio boils down to proving the following lemma. 

\begin{lemma}\label{lem:edgeweightedbound} 
There exist $\eps, \delta \in [0,1]$ such that for any time $t$ and weight $w\geq 0$, 
the variable $\hat{R}_{t,w} := \sum_{i: w_{i,t}\geq w} \hat{r}_{i,t} \cdot F_{i,t}$ satisfies 
$$\E[\min(1, \hat{R}_{t,w})] \ge \edgeapprox \cdot   \sum_{i : w_{i,t} \ge w} \frac{x_{i,t}}{p_t} .$$
\end{lemma} 

Before proving 
\Cref{lem:edgeweightedbound}, we show that it implies our main result.
\begin{proof}[Proof of \Cref{thm:edge-weighted}]
That the algorithm can be implemented in polynomial time is immediate, as \eqref{LP-PPSW} is a polynomially-sized linear program.
For the approximation ratio, let $\mathcal{M}$ denote the matching produced by \Cref{alg:proposals-edge-weighted}, and note
\begin{align*}
    \sum_t \E[w(\mathcal{M}(t)] &= \sum_t \int_0^{\infty} \Pr[\calM\ni (i,t): \; w_{i,t} \ge w] \, dw \\
    &\ge \sum_t p_t \cdot \int_0^{\infty} \E[\min(1,\hat{R}_{t,w})] \, dw && \text{\Cref{per-t-min-bound}} \\
    &\ge \sum_t p_t \cdot \int_0^{\infty} \edgeapprox \cdot \sum_{i: w_{i,t} \ge w} \frac{x_{i,t}}{p_t} \, dw  && \text{\Cref{lem:edgeweightedbound}} \\
    &= \edgeapprox \cdot \sum_t \sum_i x_{i,t} \cdot w_{i,t} \\
    & = \edgeapprox \cdot \text{OPT}(\textrm{\ref{LP-PPSW}}) && \textrm{Choice of $\vec{x}$}\\
    & \ge \edgeapprox \cdot OPT_{on}. && \textrm{\Cref{lem:LP}} \qedhere 
\end{align*}
\end{proof}

\begin{proof}[Proof of \Cref{lem:edgeweightedbound}] 

For ease of notation we fix $w$ and let $I := \{i : w_{i,t} \ge w\}$. Let $L_t := \{i\in I : y_{i,t} \le \theta \}$ and $H_t=I\setminus L_t$ denote the low- and high-degree neighbors of $t$. Note that any $i \in L_t$ has $$\hat{y}_{i,t} = \int_0^{y_{i,t}} f(z) \, dz = (1-\eps)y_{i,t} \le (1-\eps) \theta .$$ For ease of notation, let $x_L := \sum_{i\in L_t} \frac{x_{i,t}}{p_t}$ and $x_H:=\sum_{i\not\in L_t} \frac{x_{i,t}}{p_t}$. Define 
\begin{align*}\hat{x}_L &:= \sum_{i \in L_t} \frac{\hat{x}_{i,t}}{p_t} = \frac{1}{p_t} \cdot \sum_{i \in L_t} \int_{y_{i,t}}^{y_{i,t}+x_{i,t}} f(z) \, dz \ge (1-\eps) \cdot x_L, \\
\hat{x}_H & := \sum_{i \not\in L_t} \frac{\hat{x}_{i,t}}{p_t} = (1 + \delta) \cdot x_{H}.
\end{align*}

For $\hat{\theta} := \theta(1-\eps)$ we have that $\hat{y}_{i,t} \le \hat{\theta}$ if and only if $i \in L_t$. To lower bound $\E[\min(1,\hat{R}_t)]$, we will apply the fractional bucketing bound with threshold $\hat{\theta}$. For convenience of notation, we define the following function which naturally arises in this bound.

\begin{Def}
    For $\theta \in [0,1]$, define $g_{\theta}(x) := \left( 1 - (1-\theta) \cdot \Big\{ \frac{x}{1-\theta} \Big\} \right) \cdot \theta^{\lfloor \frac{x}{1-\theta} \rfloor}$.
\end{Def}

By some technical facts concerning the function $g(\cdot)$, namely Claims \ref{claim:gdecreasing} and \ref{scalingsumto1}, whose statements and proofs are included in \Cref{app:edgeweighted}, we obtain the following. 
\begin{align*}
    \E[\min(1,\hat{R}_t)] &\ge 1 - g_{\hat{\theta}} \left( \hat{x}_L \right)\cdot \prod_{i\in H_t} (1-\hat{x}_{i,t}/p_t) && \text{\Cref{lem:fracbucketingbound}} \\
    &\ge 1- g_{\hat{\theta}} \left( \hat{x}_L \right)\cdot \exp(-\hat{x}_H) &&  \hspace{-2em} 1-z\leq \exp(-z)\\
    &= 1- g_{\hat{\theta}} \left( \hat{x}_L \right)\cdot \exp(-(1+\delta) x_H)  \\
    &\ge 1- g_{\hat{\theta}} \left( (1-\eps) x_L \right)\cdot \exp(-(1+\delta) x_H) && 
    \text{ \Cref{claim:gdecreasing}} \\
    &\ge (x_L + x_H)\cdot \left( 1- g_{\hat{\theta}} \left( (1-\eps) \frac{x_L}{x_L+x_H}\right)\cdot \exp\left(-(1+\delta) \frac{x_H}{x_L+x_H}\right)\right) && \text{ \Cref{scalingsumto1}} \\
    &\ge \left( \sum_i \frac{x_{i,t}}{p_t} \right) \cdot k_{\eps, \delta} \left(\frac{x_L}{x_L+x_H}\right),
\end{align*}
for $k_{\eps,\delta}(z) :=  \Big( 1- g_{\hat{\theta}} \left( (1-\eps) \cdot z \right)\cdot \exp(-(1+\delta) (1-z)) \Big).$
With computer assistance we choose the parameters $\eps=0.11,
\delta=0.18$ that result in the (approximately) optimal lower bound 
\begin{align}\label{eqn:optimization-edge-weighted}
k_{\eps,\delta}(z) \ge \edgeapprox \qquad \forall z\in [0,1].
\end{align}
In particular, we evaluate the function in the RHS of \Cref{eqn:optimization-edge-weighted} at $10^4$ equally-spaced points $z\in [0,1]$ and rely on Lipschitzness of the RHS (see \Cref{claim:lipschitz}) to argue that the error obtained this way is at most a negligible $3 \cdot 10^{-4}$.\footnote{See code at \url{https://tinyurl.com/ydmndape}.} 
See also \Cref{fig:edge-weighted} for a pictorial validation of this bound. 
\end{proof}
\vspace{-0.3cm}

\begin{figure}[H]
\centering \includegraphics[width=0.8\textwidth]{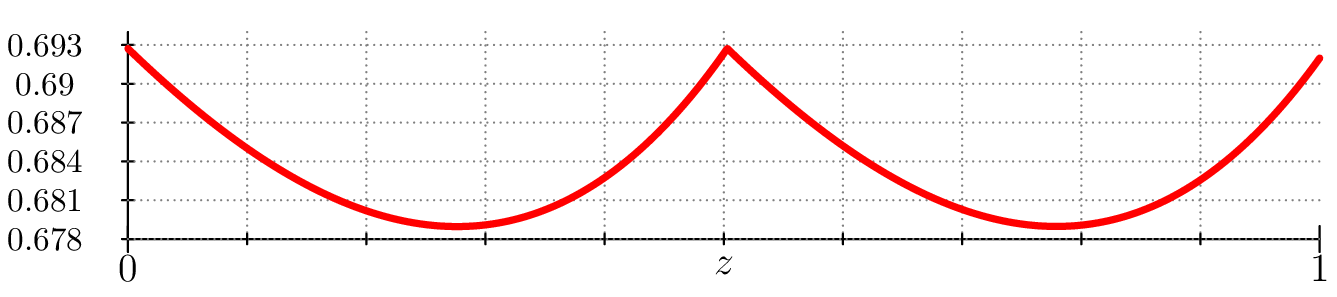}
\caption{A plot of  $k_{\eps,\delta}(z) \ge \edgeapprox$ as a function of $z\in [0,1]$.}
    \label{fig:edge-weighted}
\end{figure}

\vspace{-0.75cm}

%% file: vertexweighted.tex
\section{The Vertex-Weighted Algorithm} \label{sec:vertexweighted}

In this section we provide improved bounds for the vertex-weighted and unweighted problems. Here we avoid the use of scaling, and simply run \Cref{alg:proposals-core} on an optimal solution to \eqref{LP-PPSW}. 

\begin{algorithm}[H]
	\caption{Online Correlated Proposals Unscaled}
	\label{alg:proposals-vtx-weighted}
	\begin{algorithmic}[1]
\State Compute optimal solution $\vec{x}$ to \eqref{LP-PPSW}
\State Run \Cref{alg:proposals-core} using $\vec{x}$
\end{algorithmic}
\end{algorithm}

\paragraph{Reducing to unweighted matching.} The following lemma, which follows from the sorted order of the vector $\vec{v} = (r_{i,t})_{i\in F_t}$, allows us to avoid notational clutter and focus on unweighted graphs in our analysis, while still retaining the same guarantees for vertex-weighted matching.
\begin{lem}\label{lem:weighted-to-unweighted}
If \Cref{alg:proposals-vtx-weighted} produces a matching with expected size at least $\alpha \cdot \textup{OPT}\eqref{LP-PPSW}$ for any \emph{unweighted} graph, then it is $\alpha$-approximate for any \emph{vertex-weighted} graph. 
\end{lem}
\begin{proof}

Consider an arbitrary vertex-weighted input $G$ and threshold $w \ge 0$. 
Consider the input $G_w$ consisting of the subgraph of $G$ induced by the online nodes and only the offline nodes with weights at least $w$; with all vertex-weights set to $w$. 
By \Cref{per-t-min-bound} and the current lemma's hypothesis, 
$$\frac{\sum_t p_t \cdot \E[\min(1,R_{t,w})] \cdot w}{\sum_t \sum_i x_{i,t} \cdot \mathbbm{1}[w_i \ge w] \cdot w} \ge \alpha.$$ 
Hence the approximation ratio of \Cref{alg:proposals-core} on a vertex-weighted input is at least
\begin{align*}
\frac{\int_{w=0}^\infty \sum_t p_t \cdot \E[\min(1,R_{t,w})] \, dw}{OPT_{on}} &\geq   \frac{\int_{w=0}^\infty \sum_t p_t \cdot \E[\min(1,R_{t,w})] \, dw}{\sum_t \sum_i x_{i,t} \cdot w_i } & \textrm{\Cref{lem:LP}} \\
&\ge \frac{\int_{w=0}^\infty \alpha \cdot \left( \sum_t \sum_i x_{i,t} \cdot \mathbbm{1}[w_i \ge w] \right) \, dw}{\sum_t \sum_i x_{i,t} \cdot w_i } \\
&= \alpha \cdot \frac{\sum_t \sum_i x_{i,t} \cdot \int_{w=0}^\infty \mathbbm{1}[w_i \ge w] \, dw}{\sum_t \sum_i x_{i,t} \cdot w_i}  = \alpha.  && \qedhere
\end{align*}
\end{proof}

In light of \Cref{lem:weighted-to-unweighted}, we henceforth focus our attention on $R_t:= R_{t,0} = \sum_i r_{i,t}\cdot F_{i,t}$. 

\subsection{Reducing global to local bounds via convexity}
The bounds we obtain on $\E[\min(1,R_t)]$ are convex in various parameters associated with online nodes $t$. The following lemma allows us to leverage such convex lower bounds to obtain a lower bound on $\sum_t p_t \cdot \E[\min(1,R_t)]$, the expected size of the matching produced by \Cref{alg:proposals-vtx-weighted}.

\begin{lem}\label{lem:convexavgacrosst}
Let $f$ be a convex function and $(\boldsymbol{\gamma}_t)_t$ be a set of vectors indexed by time.
If for every time $t$ we have that
$\E[\min(1,R_t)]\geq f(\boldsymbol{\gamma}_t)\cdot \E[R_t]$, 
then \Cref{alg:proposals-vtx-weighted} is $f(\gamma)$-approximate, where $$\boldsymbol{\gamma} := \frac{\sum_t p_t \cdot \E[R_t] \cdot \boldsymbol{\gamma}_t}{\sum_t p_t \cdot \E[R_t]} = \frac{\sum_t (\sum_i x_{i,t}) \cdot \boldsymbol{\gamma}_t}{\sum_t \sum_i x_{i,t}}.$$
\end{lem}
\begin{proof}
Indeed, the approximation ratio of \Cref{alg:proposals-vtx-weighted} is lower bounded by
\begin{align*}
\frac{\E[|\calM|]}{OPT_{on}} & \geq
\frac{\sum_t p_t \cdot \E[\min(1,R_t)]}{\sum_{i,t} x_{i,t}} && \textrm{\Cref{per-t-min-bound} + \Cref{lem:LP}}
\\
&= \frac{\sum_t p_t \cdot \E[\min(1,R_t)]}{\sum_t p_t \cdot \E[R_t]} && \text{\Cref{cor:Fit}}\\
&\ge \sum_t \left( \frac{p_t \cdot \E[R_t]}{\sum_t p_t \cdot \E[R_t]} \right)  \cdot f(\boldsymbol{\gamma}_t) && \text{Lemma's Hypothesis}\\
&\ge f \left(\sum_t  \frac{p_t \cdot \E[R_t]}{\sum_t p_t \cdot \E[R_t]}  \cdot \boldsymbol{\gamma}_t \right) && \text{Jensen's inequality} \\
&= f(\boldsymbol{\gamma}). && \qedhere
\end{align*} 
\end{proof}

When applying \Cref{lem:convexavgacrosst}, we will naturally produce convex functions of certain ``weighted averages'' of our LP solution $\vec x$.

When applying the fractional bucketing bound, we will naturally categorize each edge $(i,t)$ as ``low-degree'' or ``high-degree'' based on how $y_{i,t}$ compares to $\theta$. 
\begin{Def}
    For $\theta\in [0,1]$, 
    we define the following weighted averages of $r_{i,t}\cdot y_{i,t}$, of $r_{i,t}\cdot (1-y_{i,t})$, and of $r_{i,t}$, all weighted~by~$x_{i,t}$. Some are further split for low- and high-degree edges for convenience in our future analysis:
    \begin{align*}
    \alpha & :=\frac{\sum_{i,t} r_{i,t}\cdot y_{i,t} \cdot x_{i,t}}{\sum_{i,t} x_{i,t}},\\
    \beta^{\le \theta} & := \frac{\sum_{i,t} r_{i,t } \cdot (1-y_{i,t}) \cdot x_{i,t} \cdot \mathbbm{1}[y_{i,t} 
 \le \theta]}{\sum_{i,t} x_{i,t}}, \\
 \beta^{> \theta}   & := \frac{\sum_{i,t} r_{i,t} \cdot (1-y_{i,t}) \cdot x_{i,t} \cdot \mathbbm{1}[y_{i,t} >\theta]}{\sum_{i,t} x_{i,t}}, \\
S^{\le \theta}  & := \frac{\sum_{i,t} r_{i,t } \cdot x_{i,t} \cdot \mathbbm{1}[y_{i,t} 
 \le \theta]}{\sum_{i,t} x_{i,t}}, \\
 S^{> \theta} &:= \frac{\sum_{i,t} r_{i,t}  \cdot x_{i,t} \cdot \mathbbm{1}[y_{i,t} >\theta]}{\sum_{i,t} x_{i,t}}.
\end{align*}
\end{Def}

Our bounds on the approximation ratio are decreasing and increasing in $\alpha$ and some convex function of $\beta^{>\theta},\beta^{\leq \theta}$, respectively. This motivates the following two lemmas.

 \begin{restatable}{lemma}{betaHbound}\label{lem:betaHbound}
     $\beta^{> \theta} \ge \frac{(S^{> \theta})^2}{2}$ for $\theta\in [0,1]$.
 \end{restatable}

  \begin{restatable}{lemma}{betaLbound}\label{lem:betaLbound}
    $\beta^{\le \theta} \ge S^{\le \theta} \cdot \left( 1-\theta + \frac{S^{\le \theta}}{2} \right)$ for $\theta\in [0,1]$.
 \end{restatable}

\begin{wrapfigure}[15]{R}{0.4\textwidth} 
\caption{Tight case of Lemmas \ref{lem:betaHbound}, \ref{lem:betaLbound}}
 \includegraphics[width=0.4\textwidth]{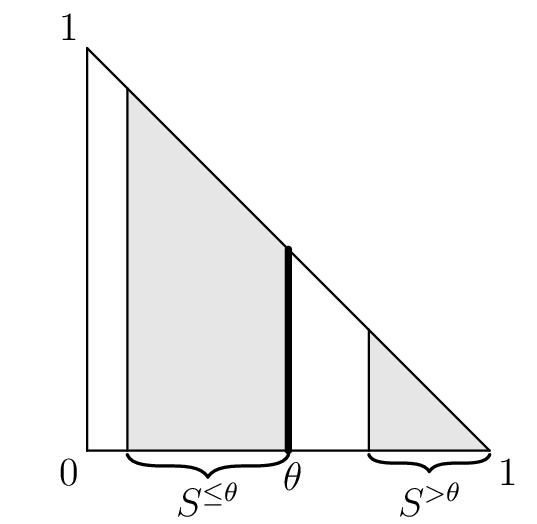}
\label{fig:betabound}
\end{wrapfigure}

 Both lemmas are proven in \Cref{app:vertexweighted}. Here, we provide a brief proof sketch. 
 Note first that the functions $\frac{x^2}{2}$ and $x \left( 1 - \theta + \frac{x}{2}\right)$ are convex; so it suffices to prove the lemma when the graph is restricted to a single offline node $i$. The worst case occurs when all $x_{i,t}$'s are infinitesimal and $\sum_t x_{i,t} = 1$ (we can ``split" a large $x_{i,t}$ into smaller $x_{i,t'}$'s without changing the value of $S^{>\theta}$ but increasing the value of $\beta^{>\theta}$). Additionally observe that the worst case for \Cref{lem:betaHbound} is when $r_{i,t}$ is 1 for the final $S^{> \theta}$ fraction of $i$'s mass, and 0 otherwise. In this case, the value of $\beta^{> \theta}$ is given by the area of the shaded triangle in \Cref{fig:betabound}. Similarly, the worst case for \Cref{lem:betaLbound} occurs when $r_{i,t}$ is 1 for the $S^{\le \theta}$ fraction of $i$'s mass from $\theta - S^{\le \theta}$ to $\theta$; in this case, the value of $\beta^{\le \theta}$ is given by the area of the shaded trapezoid in \Cref{fig:betabound}.

\subsection{Variance-based bounding }

The following lemma leverages the variance-based bound \Cref{lem:Emin1X-variance} to lower bound our approximation ratio.

\begin{lemma}\label{lem:variance-bound}
    The approximation ratio of \Cref{alg:proposals-vtx-weighted} is at least $f_{\textup{\textsf{var}}}(\alpha)$, where $f_{\textup{\textsf{var}}}(z) := 1 - \frac{1}{2}\sqrt{z}$. 
    Moreover, $\E[|\calM|]/\sum_{i,t} x_{i,t} \geq f_{\textup{\textsf{var}}}(\alpha)$.
\end{lemma}
\begin{proof}
By \Cref{obs:matchmarginals} and \eqref{LP-PPSW} constraint \eqref{eqn:OnlineMatchingConstraint}, we have that $\E[R_t]=\sum_{i} x_{i,t}/p_t\leq 1$. So, by \Cref{lem:Emin1X-variance}, we have that $\E[\min(1,R_t)] \ge f_{\textup{\textsf{var}}} \left( \frac{\Var(R_t)}{\E[R_t]}\right)\cdot\E[R_t].$ Therefore, as $f_{\textup{\textsf{var}}}$ is convex, applying \Cref{lem:convexavgacrosst} to average across $t$, we find that \Cref{alg:proposals-core}'s approximation ratio is at least $$ f_{\textup{\textsf{var}}} \left( \frac{\sum_t p_t \cdot \E[R_t] \cdot \frac{\Var(r_{i,t} \cdot F_{i,t})}{\E[R_t]} }{\sum_t \sum_i x_{i,t}} \right) = f_{\textup{\textsf{var}}} \left( \frac{\sum_t p_t \cdot \Var(R_t)}{\sum_t \sum_i x_{i,t}} \right).$$
As $f_{\textup{\textsf{var}}}(\cdot)$ is monotone decreasing 
in its argument, it suffices for us to lower bound by $\alpha$ said argument in the RHS above.
For this, we first note that since NCD variables are pairwise negatively correlated, and hence their variance is sub-additive, \Cref{lem:NCD} implies that
$$\Var(R_{t}) = \Var\left(\sum_{i} r_{i,t} \cdot F_{i,t} \right) \leq \sum_{i} \Var( r_{i,t} \cdot F_{i,t} ) = \sum_{i} r^2_{i,t} \cdot \Var( F_{i,t}).$$
We therefore have that 
\begin{align*}
\frac{\sum_t p_t \cdot \Var(R_t)}{\sum_t \sum_i x_{i,t}} &\le \frac{\sum_t p_t \cdot \sum_i \Var(r_{i,t}) \cdot F_{i,t} }{\sum_t \sum_i x_{i,t}}  \\
&= \frac{\sum_t p_t \cdot \sum_i r_{i,t}^2 (1 - y_{i,t}) y_{i,t}}{\sum_t \sum_i x_{i,t}} && F_{i,t} \sim \Ber(1-y_{i,t}) \textrm{, by \Cref{cor:Fit}} \\
&= \frac{\sum_t p_t \cdot \sum_i r_{i,t} \cdot \frac{x_{i,t}}{p_t} \cdot y_{i,t}}{\sum_t \sum_i x_{i,t}}  && \text{Def. } r_{i,t}=\frac{x_{i,t}}{p_t(1-y_{i,t})} \\
&= \alpha. && \qedhere
\end{align*}
\end{proof}
 
We briefly note that on its own, \Cref{lem:variance-bound} already implies an approximation greater than $1-1/e$. Indeed, note that $\alpha = S^{\leq 1}-\beta^{\leq 1} \leq S^{\leq1} \left( 1-\frac{S^{\leq1}}{2} \right)\leq \nicefrac{1}{2} $ by \Cref{lem:betaLbound}, and and $f_{\textup{\textsf{var}}}$ is decreasing in $[0,1]$. Thus $$ f_{\textup{\textsf{var}}}(\alpha) \ge f_{\textup{\textsf{var}}}\left( \nicefrac{1}{2} \right) = 1 - \frac{1}{2 \sqrt{2}} \approx 0.646.$$ 
 
The above bound is subsumed by our $\edgeapprox$ ratio for the more general edge-weighted problem. 
In what follows, we show how to combine the above lemma with an averaged version of the fractional bucketing bound (\Cref{lem:fracbucketingbound}) to obtain a better approximation for the vertex-weighted problem.

\subsection{Fractional-bucketing-based bounding}

In this section, we will apply the fractional bucketing bound to the unweighted matching problem. To bound the terms corresponding to high-degree offline nodes, we will use the following consequence of Hölder's inequality.

\begin{lemma}\label{lem:holder's-consequence}
If $z_1, z_2, \ldots, z_n \ge 0$ with $\sum_i z_i = S$ and $\sum_i z_i^2 = C$, then for any $k\ge 1$ we have $$\sum_{i=1}^n z_i^k \ge \frac{C^{k-1}}{S^{k-2}}.$$
\end{lemma}
\begin{proof}
Hölder's inequality states that for any vectors $\vec{u}, \vec{v} \in \mathbb{R}^n$ and $r,s > 0$ we have $$\left( \sum_{i=1}^n |u_i|^r |v_i|^s \right)^{r+s} \le \left( \sum_{i=1}^n |u_i|^{r+s} \right)^r \left( \sum_{i=1}^n |v_i|^{r+s} \right)^s.$$
Taking $\vec{u} = \left(z_i^{1+\frac{1}{k-1}} \right)_{i=1}^n$, $\vec{v} = \left(z_i^{\frac{1}{k-1}} \right)_{i=1}^n$, $r = 1$ and $s = k-2$, we obtain the desired claim, as 
\begin{align*}
\left( \sum_{i=1}^n  z_i^2 \right)^{k-1} & \le \left( \sum_{i=1}^n z_i^{k} \right) \left( \sum_{i=1}^n z_i \right)^{k-2}.\qedhere
\end{align*}
\end{proof}

The preceding consequence of H\"older's inequality implies the following bound.

\begin{lemma}\label{lem:prodoneminusz}
    For any real numbers $z_1, z_2, \ldots, z_n \in [0,1]$ with $\sum_i z_i = S$ and $\sum_i z_i^2=C$, we have $$\prod_{i=1}^n \left( 1 - z_i \right) \le \textup{exp} \left( \frac{S^2}{C} \cdot \ln (1-C/S) \right).$$
\end{lemma}

\begin{proof}
    We bound via the Taylor expansion of $\ln(1-z)$:
\begin{align*}
   \prod_{i=1}^n \left( 1- z_i\right) 
    &=  \text{exp} \left( \sum_i  \ln(1 - z_i) \right)  \\
    &=    \text{exp} \left( - \sum_{k=1}^{\infty} \sum_i  \frac{1}{k} \cdot z_i^k \right)   && \ln(1-z) = \sum_{k=1}^\infty -\frac{1}{k} z^k\\
    &\le   \text{exp} \left( - \sum_{k=1}^\infty \frac{1}{k} \cdot \frac{C^{k-1}}{S^{k-2}} \right)  && \text{\Cref{lem:holder's-consequence}} \\
    &=    \text{exp} \left( \frac{S^2}{C} \cdot \ln (1-C/S) \right). && \qedhere
\end{align*}

\end{proof}

We now apply the above lemma to the fractional bucketing bound, additionally averaging across all $t$ to get one lower bound on the approximation ratio. 
\begin{lemma}\label{lem:conv}
    For any fixed $\theta$, let $\conv_{\theta}(x,y)$ denote a convex function, increasing in both arguments, which lower bounds  $$ 1 - g_{\theta}  \left(  x\right) \cdot  \left( 1 - \frac{y}{1 - x } \right)^{\frac{(1- x)^2}{y}}$$ in the domain $\{0 \le x \le 1, \;0 \le y \le (1-x)^2\}$.\footnote{When $y=0$, this function is undefined, so formally we define $f(x,0) := \lim_{y \rightarrow 0+} f(x,y) = 1 - g_{\theta}(x) \cdot \exp(-(1-x)).$ If $x=1$, we must have $y=0$, and we also use this definition.} Then, the approximation ratio of \Cref{alg:proposals-vtx-weighted} is lower bounded by  $\conv_{\theta} \left( \theta, \beta^{> \theta} \right).$
    Moreover, $\E[|\calM|]/\sum_{i,t} x_{i,t}\geq 
 \conv_{\theta} \left( \theta, \beta^{> \theta} \right)$.
\end{lemma}

\begin{proof}
For an online node $t$, let $L_t := \{i : y_{i,t} \le \theta \}$ and $H_t := [n]\setminus L_t= \{i : y_{i,t} > \theta\}$. Additionally, let $\gamma_t := \sum_{i \in L_t} \frac{x_{i,t}}{p_t}$ and $\delta_t := \sum_{i \in H_t} \frac{x_{i,t}}{p_t}$.  We next naturally define $$\beta_t^{> \theta} := \frac{\sum_i x_{i,t} \cdot r_{i,t} \cdot (1-y_{i,t}) \cdot \mathbbm{1}[y_{i,t} > \theta]}{\sum_i x_{i,t}} $$ to be the value of $\beta^{> \theta}$ when restricted to edges incident to $t$. We briefly note that $$ \sum_{i \in H_t} \left( \frac{x_{i,t}}{p_t} \right)^2 = \frac{\sum_i x_{i,t} \cdot \frac{x_{i,t}}{p_t} \cdot \mathbbm{1}[y_{i,t} > \theta]}{p_t  } = \E[R_t] \cdot \frac{\sum_i x_{i,t} \cdot \frac{x_{i,t}}{p_t} \cdot \mathbbm{1}[y_{i,t} > \theta]}{\sum_i x_{i,t} } = (\gamma_t + \delta_t) \cdot \beta^{> \theta}_t.$$ 
    Finally, let $\Gamma_t := \frac{\sum_{i \in L_t} x_{i,t}}{\sum_{i} x_{i,t}} = \frac{\gamma_t}{\E[R_t]}$. Note
    \begin{align*}
        \frac{\E[\min(1,R_t)]}{\E[R_t]} &\ge \frac{1 - g_\theta \left( \sum_{i \in L_t} \frac{x_{i,t}}{p_t} \right) \cdot \prod_{i \in H_t} \left( 1 - \frac{x_{i,t}}{p_t} \right) }{ \sum_i \frac{x_{i,t}}{p_t}} && \textrm{\Cref{lem:fracbucketingbound}} \\
        &\ge \frac{1 - g_\theta \left( \gamma_t \right) \cdot \left(  1 - \frac{(\gamma_t + \delta_t) \cdot \beta_t^{> \theta}}{\delta_t} \right)^{\frac{\delta_t^2}{(\gamma_t + \delta_t) \cdot \beta_t^{> \theta}}} }{ \gamma_t + \delta_t} && \text{\Cref{lem:prodoneminusz}} \\
         &\ge  1 - g_{\theta}  \left( \Gamma_t \right) \cdot  \left( 1 - \frac{\beta_t^{> \theta}}{1 - \Gamma_t } \right)^{\frac{(1-\Gamma_t)^2}{\beta_t^{> \theta}}}    && \text{scaling to $\gamma_t + \delta_t = 1$ (\Cref{scalingsumto1})} \\
         &\ge \conv_{\theta}(\Gamma_t, \beta_t^{> \theta}) && \text{Def. } \conv_{\theta}(\cdot)
    \end{align*}
To conclude, we observe that if $\mathcal{M}$ denotes the matching produced by \Cref{alg:proposals-vtx-weighted}, we have
\begin{align*}
\frac{\E[|\mathcal{M}|]}{\sum_{i,t} x_{i,t}} &\ge \frac{\sum_t p_t \cdot \E[\min(1,R_t)]}{\sum_t p_t \cdot \E[R_t]} && \textrm{\Cref{per-t-min-bound}} \\
&\ge \frac{\sum_t p_t \cdot \E[R_t] \cdot \conv_{\theta}(\Gamma_t, \beta_t^{> \theta}) }{\sum_t p_t \cdot \E[R_t]} \\
&\ge \conv_{\theta} \left( \frac{\sum_t p_t \cdot \E[R_t] \cdot \Gamma_t }{\sum_t p_t \cdot \E[R_t]}, \frac{\sum_t p_t \cdot \E[R_t] \cdot \beta_t^{>\theta}}{\sum_t p_t \cdot \E[R_t]}\right) && \text{\Cref{lem:convexavgacrosst}}  \\
&\ge \conv_{\theta} \left( \theta, \beta^{>\theta} \right).
\end{align*}
In the final inequality, we used the observations that $$\frac{\sum_t p_t \cdot \E[R_t] \cdot \Gamma_t }{\sum_t p_t \cdot \E[R_t]} = \frac{\sum_i \sum_t x_{i,t} \cdot \mathbbm{1}[y_{i,t} \le \theta]}{\sum_t \sum_i x_{i,t}} \ge \theta $$ and 
\begin{align*}
\frac{\sum_t p_t \cdot \E[R_t] \cdot \beta^{> \theta}_t }{\sum_t p_t \cdot \E[R_t]} &=  \beta^{>\theta}. \qedhere
\end{align*} 
\end{proof}

\subsection{Combining the variance bound with fractional bucketing}

In this section we provide a unified analysis applying both the variance bound and the fractional bucketing bound, yielding our main result of this section.

\begin{thm}\label{thm:vertex-weighted}
    \Cref{alg:proposals-vtx-weighted} is a polynomial-time $\vertexapprox$-approximate online algorithm for vertex-weighted (or unweighted) online stochastic bipartite matching.   
\end{thm}
\begin{proof}
    That the algorithm runs in polynomial time is immediate; the remainder of this proof is dedicated to bounding the approximation ratio.

    By \Cref{lem:weighted-to-unweighted}, it suffices to prove that $\E[|\calM|]/\sum_{i,t} x_{i,t} \geq \vertexapprox$ when the algorithm is run on any unweighted instance.
   Now, by Lemmas~\ref{lem:variance-bound} and \ref{lem:conv}, we have 
    \begin{align}
    \label{eqn:first-vtxweighted-step} \frac{\E[|\calM|]}{\sum_{i,t} x_{i,t}} 
    \geq \max\left\{1-\frac{1}{2}\sqrt{\alpha},\; \conv_{\theta}(\theta, \beta^{>\theta})\right\},
    \end{align}
    for any convex function $\conv_{\theta}(x,y)\leq 1 - g_{\theta}  \left(  x\right) \cdot  \left( 1 - \frac{y}{1 - x } \right)^{{(1- x)^2/y}}$ increasing in both arguments.
This can in turn be lower bounded as an expression in two arguments, $S^{\leq \theta}, S^{>\theta}$, by using Lemmas~\ref{lem:betaHbound} and \ref{lem:betaLbound}, which imply that $\beta^{>\theta}\geq \frac{(S^{>\theta})^2}{2}$ and also that 
$$\alpha = S^{\le \theta} + S^{> \theta} - \beta^{\le \theta} - \beta^{> \theta} \le S^{\le \theta} + S^{> \theta} - S^{\le \theta} \cdot \left( 1-\theta - \frac{S^{\le \theta}}{2} \right) -  \frac{(S^{> \theta})^2}{2}.$$
Using that $\conv_{\theta}(\cdot, \cdot)$ is monotone increasing in its arguments, while $1- \frac{1}{2}\sqrt{x}$ is decreasing,  we have that the approximation ratio of \Cref{alg:proposals-vtx-weighted} on vertex-weighted instances is at least
$$(\dagger)_{\theta} := \min_{S^{\le \theta}, S^{> \theta}} \max \left( 1 - \frac{1}{2} \sqrt{ S^{\le \theta} + S^{> \theta} - S^{\le \theta} \cdot \left( 1-\theta + \frac{S^{\le \theta}}{2} \right) -  \frac{(S^{> \theta})^2}{2} },  \conv_{\theta} \left( \theta, \frac{(S^{> \theta})^2}{2}\right)  \right).$$
In \Cref{app:vertexweighted} we fix $\theta= \nicefrac{1}{2}$ and provide a computer-assisted proof to justify a linear lower bound $\conv_{\theta} (\cdot,\cdot)$, using Lipschitzness of $1 - g_{\theta}  \left(  x\right) \cdot  \left( 1 - \frac{y}{1 - x } \right)^{{(1- x)^2/y}}$ and evaluating a fine grid of points to bound the error in this linear lower bound. 
In similar fashion, we lower bound $(\dag)_{\theta} \geq \vertexapprox$
for $\theta= \nicefrac{1}{2}$. The theorem follows.
\end{proof}

\begin{rem}
    Our choice of $\theta= \nicefrac{1}{2}$ and linear lower bound for $\conv_{\theta} \left( \theta, \frac{(S^{> \theta})^2}{2} \right)$ may seem~arbitrary. However, computer-assisted proof optimizing over our arguments, and taking $\conv_{\theta} \left(x,y \right)$ to be the lower convex envelope of $1-g_{\theta}(x)\cdot \left(1-\frac{y}{1-x}\right)^{(1-x)^2/y}$ (approximately evaluated based on a fine grid), shows that these choices are optimal for our arguments, up to a negligible error term.
\end{rem}

%% file: hardness.tex
\section{Hardness}\label{sec:hardness}

In this section we prove the following hardness of approximation of the philosopher for the simplest, unweighted, version our problem.

\begin{thm}
It is $\mathsf{PSPACE}$-hard to approximate $OPT_{on}$ within some universal constant $\alpha < 1$, 
even for unweighted Bernoulli instances with arrival probabilities bounded away from zero. 
\end{thm}

To prove our $\mathsf{PSPACE}$-hardness result, we make a connection with hardness of approximation for the bounded-degree stochastic SAT problem as in \cite{papadimitriou2021online}; our novel contribution is showing that such a reduction works even without edge weights or small arrival probabilities.

\begin{definition}[cf. \cite{papadimitriou1985games}]
    An instance of the $\textsc{Stochastic-3-Sat}$ problem consists of a 3-CNF $\phi$ with variables $x_1, x_2, \ldots, x_n$. At time $t = 1, 2, \ldots, n$, if $t$ is odd our algorithm $\mathcal{A}$ chooses whether to set $x_t$ to $\texttt{True}$ or $\texttt{False}$ while if $t$ is even, nature sets $x_t$ uniformly at random from $\{\texttt{True}, \texttt{False}\}$. The performance of our algorithm, $\mathcal{A}(\phi)$, is the expected number of satisfied clauses. 
\end{definition}

\begin{definition}[cf. \cite{papadimitriou2021online}]
    An instance of the $k\textsc{-Stochastic-3-Sat}$ problem consists of a \textsc{Stochastic-3-Sat} instance where each variable appears in at most $k$ clauses, and every randomly set variable $x_{2i}$ never appears negated in any clause of $\phi$. 
\end{definition}

We rely on the following hardness result. 

\begin{prop}[cf. \cite{papadimitriou2021online}] There exists a positive integer $k$ and absolute constant $\alpha < 1$ such that it is $\textup{\textsf{PSPACE}}$-hard to approximate  $k\textup{\textsc{-Stochastic-3-Sat}}$ within a factor $1-\alpha$. 
\end{prop}

Consider a $k\textsc{-Stochastic-3-Sat}$ instance $\phi$ with variables $\{x_1, x_2, \ldots, x_n\}$. Let $\mathcal{C}$ denote the clauses of $\phi$, and define $C := |\mathcal{C}|$. Additionally denote $\textsc{Opt} := OPT_{on}(\phi)$. We construct a corresponding online matching instance $\mathcal{I}_{\phi}$ which first has $n$ arrivals called ``variable nodes'' corresponding to $\{x_1, \ldots, x_n\}$. For odd $t \in [n]$, the corresponding variable node arrives with probability 1, and neighbors two offline nodes $\{T^t, F^t\}$. For even $t \in [n]$, the corresponding variable node arrives with probability $\nicefrac{1}{2}$, and neighbors exactly one offline node $\{F^t\}$.

The instance $\mathcal{I}_{\phi}$ then consists of stochastic ``clause nodes" all arriving with a common probability $p \le 0.1$. For each clause $C \in \mathcal{C}$, there is a corresponding stochastic arrival of degree at most 3. The neighborhood of the stochastic arrival corresponding to $C$ is constructed as follows: for each of the literals $l_i \in C$, if $l_i = x_i$ we add an edge to $F^i$ and if $l_i = \overline{x_i}$ we add an edge to $T_i$. (Note that this is always possible by our assumption that $\phi$ has no randomly set variable appear negated in any clause.) The stochastic arrivals are ordered in an arbitrary fashion (subject to all coming after the deterministic arrivals). The instance $\mathcal{I}_{\phi}$ is depicted in the following figure. 

\begin{figure}[h] 
    \centering
    \includegraphics[width=0.95\textwidth]{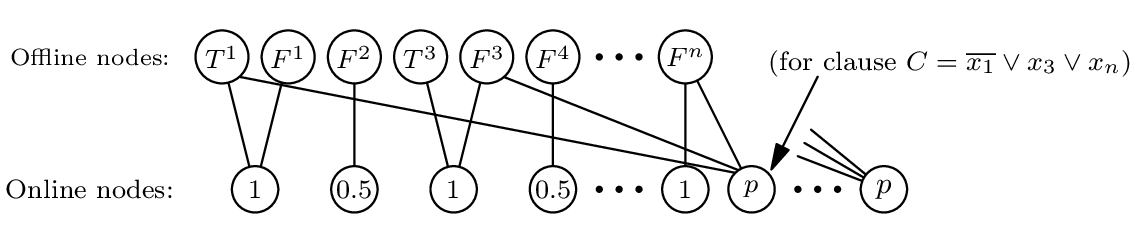}
    \caption{The instance $\calI_{\phi}$ for our $\textsf{PSPACE}$-hardness reduction }
    \vspace{0.2cm}
     \footnotesize{Bins are labeled by their corresponding literal, while balls are labeled by their arrival probability $p_t$.}
     \vspace{0.2cm}
    \label{fig:binreduction}
\end{figure}

There is a natural bijection between algorithms for the $\textsc{Stochastic-3-Sat}$ instance $\phi$, and matching algorithms for $\mathcal{I}_{\phi}$ which match every arriving variable node. As optimum online is one such algorithm, in \Cref{app:hardness} we are able to bound its performance on $\mathcal{I}_{\phi}$ in terms of $\textsc{Opt}$. Our gap in these bounds is due to the fact that multiple clause nodes could arrive, instead of one chosen uniformly at random. However, because offline nodes have low degrees, we show the amount of noise this introduces is of order $p^2$, while the signal that relates $\opton(\mathcal{I}_{\phi})$ to $\textsc{Opt}$ is of order $p$. Hence taking $p$ to be a sufficiently small constant suffices for our hardness of approximation result.

%% file: app-prelims.tex
\section{Deferred Proofs of Section~\ref{sec:prelims}: \texorpdfstring{\eqref{LP-PPSW}}{The LP}}\label{app:prelims}

Another valid constraint for \eqref{LP-PPSW} is the \emph{offline degree constraint}, $\sum_{t} x_{i, t} \le 1$ for all $i\in \OFF$. However, this is subsumed by Constraint \eqref{eqn:PPSWConstraint}, as we now show, by 
generalizing a proof of \cite{torrico2022dynamic} for the special case $p_t=1/n$~for~all~$t$.

\begin{rem}\label{lem:degree-implicit} 
Constraint \eqref{eqn:PPSWConstraint}  implies that $\sum_{t'\leq t}x_{i,t'}\leq 1-\prod_{t'<t} (1-p_{t'}) \leq 1$ for all $i,t$.
\end{rem}
\begin{proof}
We generalize the proof of \cite{torrico2022dynamic} for the special case of $p_t=1/n$ for all $t$, as follows. 
By Constraint \eqref{eqn:PPSWConstraint}, $y_{i,t}=\sum_{t'<t}x_{i,t'}$ satisfies $$y_{i,t+1}=x_{i,t}+y_{i,t}\leq p_t(1-y_{i,t})+y_{i,t}.$$ Therefore,
$$1-y_{i,t+1}\geq (1-y_{i,t})(1-p_t).$$
Concatenating the above inequality for different $t$, we obtain that 
$$1-y_{i,t}\geq (1-y_{i,t-1})(1-p_{t-2})\geq (1-y_{i,t-2})(1-p_{t-2})(1-p_{t-3})\geq \dots \geq \prod_{t'< t} (1-p_{t'}).$$
Consequently, $y_{i,t} \leq 1-\prod_{t'<t} (1-p_{t'}),$ as claimed.
\end{proof}

\begin{rem}\label{rem:near-binary}
    The $nT$-dimensional \eqref{LP-PPSW} has $2nT+T$ constraints, $2nT$ of which are of the form $0\leq r_{i,t}$ or $r_{i,t}\leq 1$, for $r_{i,t}:=x_{i,t}/(1-\sum_{t'<t}x_{i,t'})$. Hence, in any basic feasible solution of \eqref{LP-PPSW} (having at least $nT$ tight constraints), 
    the vector $\vec{r}=(r_{i,t})_{i,t}$ 
    is ``near-binary'', having at most $T$ fractional entries; i.e., at most one $r_{i,t}$ is non-binary per time step, \emph{on average}.
\end{rem}

%% file: app-algo.tex
\section{Deferred Proofs of Section~\ref{sec:algo}: The Algorithm}\label{app:algo}

\obsmatchmarginals*
\begin{proof}
Define $O_{i,t}:=F_{i,t}-F_{i,t+1}$, i.e., the indicator for $i$ being matched or discarded at time $t$. Note $\sum_t O_{i,t}\leq 1$, by definition. Then, it suffices to show $O_{i,t}$ satisfies $\Pr[O_{i,t}]=x_{i,t}$, and hence $$\Pr[F_{i,t}]=1-\sum_{t'<t}\Pr[O_{i,t'}]=1-\sum_{t'<t}x_{i,t'} = 1-y_{i,t}.$$

We induct on $t$. Note that for $i$ to be occupied at time $t$ and not previously, it must hold that $i\in I_t \cap F_t$.
Taking total probability over histories $\calH$ implying $F_{i,t}$, 
and using Property \ref{level-set:marginals} of \Cref{lem:SR}, we have that
$$\Pr[i \in I_t \mid F_{i,t}] = \sum_{\calH \to F_{i,t}} \Pr[i\in I_t \mid \calH] \cdot \Pr[\calH \mid F_{i,t}] = \sum_{\calH \to F_{i,t}} r_{i,t} \cdot \Pr[\calH \mid F_{i,t}] = r_{i,t}.$$

Next, conditioned on $i \in I_t$, the probability that $t$ is occupied in timestep $t$ is precisely $p_t$: if $i = i^*_t$, then $i$ is matched if $t$ arrives, which occurs with probability $p_t$, and if $i \neq i^*_t$, then $i$ is discarded with probability $p_t$. The inductive step then follows, by definition of $r_{i,t}=\frac{x_{i,t}}{p_t\cdot \left(1-\sum_{t'<t}x_{i,t'}\right)}$. 
\begin{align*}
\Pr[O_{i,t}] & = 
\Pr[O_{i,t} \mid i\in I_t] \cdot \Pr[i\in I_t \mid F_{i,t}] \cdot \Pr[F_{i,t}] = 
p_t \cdot r_{i,t} \cdot \left( 1 - \sum_{t' < t} x_{i,t'} \right) = x_{i,t}.\qedhere
\end{align*}

\end{proof}

\ncd*
To prove \Cref{lem:NCD}, we use the following fact from \cite{braverman2022max}, for which we provide a simple probabilistic proof for completeness. (We note that this is a special case of \Cref{fact:independent-buckets}). 

\begin{fact}\label{fact:independent-coins}
If $X$ is a random subset of $[n]$, and $0 \le p_1, p_2, \ldots, p_n \le 1$, then $$\sum_{J \subseteq [n]} \left( \Pr[X = J] \cdot \prod_{i \in J} (1 - p_i) \right) = \sum_{J \subseteq [n]} \left( \Pr[X \subseteq J]\cdot \prod_{i \in J} (1 - p_i) \cdot \prod_{i \not\in J} p_i  \right).$$
\end{fact}

\begin{proof}
Consider having each element $i \in [n]$ independently toss a biased coin which lands heads with probability $p_i$, and then realizing $X$. Both the LHS and RHS compute the probability that every element in $X$ flipped tails. 
\end{proof}

\begin{proof}[Proof of \Cref{lem:NCD}]
We apply induction on $t$; assume these inequalities hold for $t$. Fix some $I \subseteq \OFF$. Note that all nodes in $I$ are free at time $t+1$ if and only if they were all free at time $t$, and every node in $I$ that proposed at time $t$ was neither matched (in \Cref{line:match-end}) nor discarded (in \Cref{line:loop-end}), i.e., every such node failed an independent $\Ber(p_t)$ coin flip. Therefore, by summing over every subset $J\subseteq I$ that could propose (i.e., letting $I_t\cap I=J$), we have, using \Cref{fact:independent-coins}, that
\begin{align}
    \Pr \left[ \bigwedge_{i \in I} F_{i,t+1} \right] &=  \Pr \left[ \bigwedge_{i \in I} F_{i,t} \right] \cdot \sum_{J \subseteq I} \Pr\left[I_t \cap I = J \bigmid \bigwedge_{i \in I} F_{i,t}\right] \cdot (1 - p_t)^{|J|}  \nonumber \\ 
     &= \Pr \left[ \bigwedge_{i \in I} F_{i,t} \right] \cdot \sum_{J \subseteq I}  \Pr \left[ I_t\cap I \subseteq J  \bigmid \bigwedge_{i \in I} F_{i,t} \right]  \cdot (1 - p_t)^{|J|} \cdot p_t^{|I| - |J|} . \label{eqn:midinductionncd}
\end{align}

Next, observe that for any fixed history $\mathcal{H}$ up to time $t$, we have that $\{ [i \in I_t \mid \mathcal{H} ]\}$ are negative cylinder dependent, by Property \ref{level-set:neg-corr} of \Cref{lem:SR}, and so
$$\Pr \left[ I_t\cap I\subseteq J  \mid \mathcal{H} \right] = \Pr\left[\bigwedge_{i\in I\setminus J} (i\notin I_t) \bigmid \mathcal{H}\right] \le \prod_{i \in I \setminus J} \Pr\left[i \notin I_t \bigmid \mathcal{H}\right] = \prod_{i \in I \setminus J} (1 - r_{i,t}),$$ 
where in the final equality we used Property \ref{level-set:marginals} of \Cref{lem:SR}. Therefore, by total probability, 
$$ \Pr \left[ I_t\cap I\subseteq J  \bigmid \bigwedge_{i \in I} F_{i,t} \right] \le \prod_{i \in I \setminus J} (1 - r_{i,t}).$$

The above lets us simplify \eqref{eqn:midinductionncd} to prove the desired negative upper cylinder dependence.
\begin{align*}
    \Pr \left[ \bigwedge_{i \in I} F_{i,t+1} \right] &\le \Pr \left[ \bigwedge_{i \in I} F_{i,t} \right] \cdot \sum_{J \subseteq I} \prod_{i \in I \setminus J} (1 - r_{i,t}) \cdot (1 - p_t)^{|J|} \cdot p_t^{|I| - |J|}  \\
    &= \Pr \left[ \bigwedge_{i \in I} F_{i,t} \right] \cdot \prod_{i \in I} \Big( (1 - p_t) + p_t \cdot (1 - r_{i,t}) \Big) && \textrm{multi-binomial theorem} \\
    &= \Pr \left[ \bigwedge_{i \in I} F_{i,t} \right] \cdot \prod_{i \in I} \Big( 1 - \frac{x_{i,t}}{1 - y_{i,t}} \Big) && \textrm{Def. $r_{i,t}$} \\
    &\le \prod_{i \in I} \Pr \left[ F_{i,t} \right] \cdot \prod_{i \in I} \Big(\frac{1-y_{i,t+1}}{1 - y_{i,t}} \Big) && \text{I.H. + Def. $y_{i,t}$} \\
    &= \prod_{i \in I} \Pr[F_{i,t+1}]. && \text{\Cref{cor:Fit}}
\end{align*}

The upper bound on the probability of $\bigwedge_{i \in I} \overline{F_{i,t}}$ similarly inducts on $t$; it proceeds very similarly with the exception that it requires upper bounding the probability that a group of nodes all propose (as opposed to a group all not proposing), which again we can upper bound using properties \ref{level-set:marginals} and \ref{level-set:neg-corr} of \Cref{lem:SR}. For the sake of completeness, the details are included below. 

We apply induction on $t$; assume this holds for any $t' \le t$. Fix some $I \subseteq \OFF$. Note for $\wedge_{i \in I} \overline{F_{i,t+1}}$ to hold, every free node in $I$ at time $t$ must be matched or discarded. This happens if and only if they all propose and each pass a $\Ber(p_t)$ coin flip. By total probability,

$$ \Pr \left[ \bigwedge_{i \in I} \overline{F_{i,t+1}} \right] = \sum_{J \subseteq I} \Pr \left[ F_t \cap I = J \right] \cdot \Pr[I_t \cap I \supseteq J \mid F_t \cap I = J] \cdot p_t^{|J|}.$$

For any fixed $J \subseteq I$, and fixed history $\mathcal{H}$ up to time $t$ which implies $F_t \cap I = J$, because our proposals are negative cylinder dependent (Property \ref{level-set:neg-corr} of \Cref{lem:SR}) we have that
$$  \Pr[I_t \cap I \supseteq J \mid \mathcal{H} ] = \Pr \left[  \bigwedge_{j \in J} j \in I_t \;\middle\vert\; \mathcal{H} \right] \le \prod_{j \in J} \Pr[ j \in I_t \mid \mathcal{H}] = \prod_{j \in J} r_{j,t}.$$ Note here we use the opposite direction of negative cylinder dependence from that used in the first part of the lemma. Hence, by total probability, we can see that 
$$\Pr[ I_t \cap I \supseteq J \mid F_t \cap I = J] \le \prod_{j \in J} r_{j,t}.$$
Similar to before, we obtain the second upper bound.
\begin{align*}
    & \Pr \left[ \bigwedge_{i \in I} \overline{F_{i,t+1}} \right]\\
    \leq & \sum_{J \subseteq I} \Pr \left[ F_t \cap I = J \right] \cdot \prod_{j \in J} (p_t \cdot r_{j,t}) \\
    =& \sum_{J \subseteq I} \Pr \left[ F_t \cap I \subseteq J \right] \cdot \prod_{j \in J} (p_t \cdot r_{j,t}) \cdot \prod_{j \in I \setminus J} (1 - p_t \cdot r_{j,t}) && \text{\Cref{fact:independent-coins}} \\
    =& \sum_{J \subseteq I} \Pr \left[ \bigwedge_{j \in I \setminus J} \overline{F_{j,t}}  \right] \cdot \prod_{j \in J} (p_t \cdot r_{j,t}) \cdot \prod_{j \in I \setminus J} (1 - p_t \cdot r_{j,t}) && \\
    \le& \sum_{J \subseteq I} \left(\prod_{j \in I \setminus J} \Pr[\overline{F_{j,t}}] \right) \cdot \prod_{j \in J} (p_t \cdot r_{j,t}) \cdot \prod_{j \in I \setminus J} (1 - p_t \cdot r_{j,t}) && \text{I.H.} \\
    =& \sum_{J \subseteq I} \left(\prod_{j \in I \setminus J} y_{j,t}  \cdot \left(1 - \frac{x_{j,t}}{1 - y_{j,t}} \right)  \right) \cdot \prod_{j \in J} \left( \frac{x_{j,t}}{1 - y_{j,t}} \right)  &&  \text{\Cref{obs:matchmarginals}} \\
    =& \prod_{j \in I} \left( y_{j,t} \cdot \left(1 - \frac{x_{j,t}}{1 - y_{j,t}} \right)  +  \left( \frac{x_{j,t}}{1 - y_{j,t}} \right) \right)  && \textrm{multi-binomial theorem}  \\
    =& \prod_{j \in I} \left( y_{j,t} + x_{j,t} \right) \\
    =& \prod_{j \in I} y_{j,t+1} \\
     =&  \prod_{j \in I} \Pr[\overline{F_{i,t+1}}]. && \textrm{\Cref{cor:Fit}} \qedhere
\end{align*}
\end{proof}

%% file: app-Emin1X.tex
\section{Deferred Proofs of Section~\ref{sec:Emin1X}: Bounds on \texorpdfstring{$\E[\min(1,X)]$}{E[min(1,X)]}}\label{app:Emin1X}

\begin{rem}\label{claim:indprop1-1/e}
\Cref{lem:basic-independent-proposal-bound} implies that \Cref{alg:proposals-core} run on $\textup{OPT}(\textup{\ref{LP-PPSW}})$ has a $(1-1/e)$-approximation ratio, but no better.
\end{rem}

\begin{proof}
    Note that at every timestep $t$, we have 
    \begin{align*}
    \frac{\E[\min(1,R_t)]}{\E[R_t]} &\ge \frac{ 1-\prod_{i} \left(1-x_{i,t}/p_t \right)}{\sum_i x_{i,t}/p_t} \\
    &\ge \frac{ 1-\prod_{i } \exp(-x_{i,t}/p_t)}{\sum_i x_{i,t}/p_t} \\
    &= \frac{ 1- \exp(- \sum_i x_{i,t}/p_t)}{\sum_i x_{i,t}/p_t} \\
    &\ge 1-1/e. && \text{concavity of } 1-\exp(-x)
     \end{align*}
     So, by applying \Cref{lem:convexavgacrosst}, we have that \Cref{alg:proposals-core} is $(1-1/e)$-approximate. \Cref{lem:basic-independent-proposal-bound} cannot prove a better approximation ratio by itself; consider an instance with $n$ offline nodes, a single arrival at $t=1$ with probability 1 neighboring all offline nodes, and $x_{i,1} = 1/n$ for every $i \in [n]$.
\end{proof}

\subsection{The bucketing bound}

\indproposalsbucketed*

We first generalize our probabilistic facts about real numbers $0 \le c_1, c_2, \ldots c_n \le 1$ when they are grouped into a collection of ``buckets'' $\mathcal{B}$ partitioning $[n]$ such that $\sum_{i \in B} c_i \le 1$ for every bucket $B \in \mathcal{B}$. We note that our results in the paper body can be recaptured by taking singleton buckets. 

\begin{restatable}{fact}{unionboundbuckets} \label{fact:unionbound-buckets}
If $0 \le c_1, c_2, \ldots c_n \le 1$, then for any partition $\calB$ of $[n]$ into buckets such that $\sum_{i \in B} c_i \le 1$ for every $B \in \mathcal{B}$, we have that $$\min \left( 1, \,\sum_{i=1}^n c_i \right) \ge 1 - \prod_{B \in \calB} \left( 1 - \sum_{i \in B } c_i \right).$$
\end{restatable}

\begin{proof}
    For each bucket $B \in \mathcal{B}$, using that $\sum_{i \in B} c_i \le 1$, we independently pick at most one index $i \in B$ such that $\Pr[i \text{ realized}] = c_i$. 
    By independence across buckets, the RHS is exactly the probability that at least one index is picked, which is upper bounded by the LHS, 
    by the~union~bound. 
\end{proof}

\begin{restatable}{fact}{independentbuckets}\label{fact:independent-buckets}
Let $S$ be a random subset of $[n]$ and $0 \le c_1, c_2, \ldots, c_n \le 1$. Additionally, let $\calB$ be a partition of $[n]$ such that each bucket $B \in \calB$ satisfies $\sum_{i \in B} c_i \le 1$. Then, $$\sum_{J \subseteq [n]} \Pr[S = J] \cdot \prod_{B \in \calB} \left( 1 - \sum_{i \in B \cap J} c_i \right) = \sum_{\substack{J \subseteq [n]\\|B\setminus J|\leq 1 \;\forall B\in \calB}} \Pr[S \subseteq J]\cdot \prod_{i\not\in J} c_i\cdot \prod_{B\subseteq J} \left(1-\sum_{i\in B}  c_i\right).$$
\end{restatable}

\begin{proof}
Again, for each bucket $B \in \mathcal{B}$ independently, we  pick at most one index $i \in B$ such that $\Pr[i \text{ realized}] = c_i$ (using that $\sum_{i \in B} c_i \le 1$). 
The LHS and RHS both count the probability that every element in $S$ is inactive. 
In particular, the RHS expresses that $S$ is contained in the set of inactive elements $J$: either all elements in bucket $B$ belong to $J$ and $B$ therefore has no active element (w.p. $1-\sum_{i\in B}c_i$) or all elements in bucket $B$ except for a single active element $i$ (w.p. $c_i$), and these choices for active $B\cap J$ are independent across buckets $B\in \calB$.
\end{proof}

We are now ready to generalize our proof of \Cref{lem:basic-independent-proposal-bound} and thus prove \Cref{lem:bucketed-independent-proposal-bound}.
\begin{proof}
[Proof of \Cref{lem:bucketed-independent-proposal-bound}] 
Let $S:=\{i \mid X_i=1\}$ and $T:=\{J\subseteq \OFF \mid \,|B\setminus J|\leq 1\,\,\forall B\in \calB\}$. Then,
\begin{align}
\E[\min(1,X)] &  = \sum_{J \subseteq \OFF} \Pr[S = J] \cdot \min \left(1,\, \sum_{i \in J} c_i \right)  \nonumber \\
    &\ge \sum_{J \subseteq \OFF} \Pr[S = J] \cdot \left( 1 - \prod_{B \in \calB} \left( 1 - \sum_{i \in J \cap B} c_i \right) \right) && \text{\Cref{fact:unionbound-buckets}} \nonumber \\
    &=  1 - \sum_{J \subseteq \OFF} \Pr[S = J] \cdot  \prod_{B \in \calB} \left( 1 - \sum_{i \in J \cap B} c_i \right) \nonumber \\
    &= 1 - \sum_{J\in T} \left( \Pr[S \subseteq J] \cdot \prod_{i\not\in J}  c_i \cdot  \prod_{B\subseteq J} \left(1 - \sum_{i\in B} c_i\right) \right)  && \text{\Cref{fact:independent-buckets}} \nonumber \\
    & \geq 1 - \sum_{J\in T} \left(\prod_{i\not\in J} (1-q_i)\cdot c_i \right) \cdot \prod_{B\subseteq J} \left(1 - \sum_{i\in B} c_i \right)   &&  \{X_i\sim \Ber(q_i)\}\; \text{ NCD} \label{line:bucketingNCD} \\ 
    & =  1 - \sum_{\calB' \subseteq \calB} \left( \prod_{B \in \calB\setminus \calB'} \left( \sum_{i \in B} (1-q_i)\cdot c_i \right) \cdot \prod_{B \in \calB'} \left(1 - \sum_{i\in B} c_i  \right)\right) && \calB' := \{ B : B \subseteq J \}  \nonumber \\ 
    & =  1 - \prod_{B\in \calB} \left(1 - \sum_{i\in B} c_i  + \sum_{i \in B} c_i \cdot (1-q_{i}) \right)  \nonumber &&  \textrm{multi-binomial theorem} \\ 
    & =  1 - \prod_{B\in \calB} \left(1 - \sum_{i\in B} c_i\cdot q_i\right). \nonumber  &&  \qedhere 
    \end{align}
\end{proof}

\subsection{The fractional bucketing bound}

\amgmforbucketing*
\begin{proof}
    First, if $q_i = 1 - \theta$ for all $i$, and $w < 1$, then $\frac{\mu}{1-\theta} = w + |A|$, so $|A| = \lfloor \frac{\mu}{1-\theta} \rfloor$ and $w = \Big\{ \frac{\mu}{1-\theta} \Big\}.$ Hence, the inequality holds with equality.  

Next, we prove the inequality for $w = 1$ but arbitrary $\{q_i\}_{i\in A}$; note that by relabeling, it suffices to prove the inequality for $c=0$. Define $q := \sum_{i \in A} q_i/|A|$ and note that $q \ge 1-\theta$ and that $\frac{\ln(x)}{1-x}$ is increasing for $x \in (0,1)$. Thus, $\frac{\ln(1-q)}{q} \le \frac{\ln(\theta)}{1-\theta}$, implying that $(1-q)^{1-\theta} \le \theta^{q}$. We therefore have
\begin{align*}
     \prod_{i \in A} (1 - q_i) &\le (1-q)^{|A|} && \text{AM-GM} \\
    &\le  \theta^{\frac{|A|q}{1-\theta}} &&  (1-q)^{1-\theta} \le \theta^{q} \\
    &=  \theta^{ \frac{\mu}{1-\theta} } \\
    &= \theta^{ \lfloor \frac{\mu}{1-\theta} \rfloor } \cdot \theta^{\{ \frac{\mu}{1-\theta} \}} \\
    &\le \theta^{ \lfloor \frac{\mu}{1-\theta} \rfloor } \cdot \left(1 - (1-\theta) \Big\{ \frac{\mu}{1-\theta} \Big\}  \right), && x^a \le 1 - (1-x)a \text{ for } x,a \in [0,1]
\end{align*} 
which concludes the proof for $c = 0$, and hence for $c=1$.

    Suppose then that there exists some index $j \in A$ such that $q_j > 1 - \theta$ and $c\neq 1$. 
    We transform our instance $(\{q_i\}_{i \in A}, c)$ into a new instance $(\{q_i'\}_{i \in B}, c')$ which (I) preserves $\mu$, i.e., $\mu(c',\vec{q}')=\mu(c,\vec{q})$, thus preserving the RHS of the desired inequality, and (II) does not increase the LHS of the desired inequality,  and, crucially, (III) is of one of the preceding simpler forms for which the inequality was proven above. This then implies the desired inequality.
    $$(1 - c\cdot(1 - \theta)) \cdot \prod_{i \in A} (1 - q_i) \le (1 - c'(1 - \theta)) \cdot \prod_{i \in A} (1 - q_i')\leq \left( 1 - (1-\theta) \cdot \Big\{ \frac{\mu}{1-\theta} \Big\} \right) \cdot \theta^{\lfloor \frac{\mu}{1-\theta} \rfloor}.$$
    
    The transformation is as follows:
    set $q_j' = q_j - \epsilon$ for some $\epsilon = \min(q_j - (1-\theta), (1-c)(1-\theta))$. For all $i \in A \setminus \{j\}$, set $q_i' = q_i$, and set $c' = c + \frac{\epsilon}{1-\theta}.$ Note that $\mu$ is preserved under this transformation. Note additionally that by our choice of $\epsilon$, $q_j' \ge 1-\theta$ and $c' \le 1$, and at least one of these inequalities is tight. We will show $$(1 - c\cdot(1 - \theta)) \cdot \prod_{i \in A} (1 - q_i) \le (1 - c'(1 - \theta)) \cdot \prod_{i \in A} (1 - q_i').$$
    After this inequality is established, we can successively apply these transformations until (i) $c=1$, or (ii) every $q_i = 1 - \theta$, and appeal to our previous analysis for the relevant case. Our desired inequality is equivalent to 
    $$(1 - c\cdot(1 - \theta)) \cdot (1-q_j) \le \left( 1 - c\cdot(1-\theta) + \eps \right) \cdot \left( 1- q_j + \epsilon \right),$$ 
    which holds since $(1-a)\cdot (1-b) \leq (1-a-\eps)\cdot (1-b + \eps)$ for real $a,b$ satisfying $\eps\leq b-a$, and indeed in our case $a=c\cdot(1-\theta)$ and $b=q_j$ satisfy $\eps\leq q_j - (1-\theta) \leq q_j- c\cdot(1-\theta)$, as $c\leq 1$. 
\end{proof}

\subsection{The variance bound}
\varbound*
\begin{proof}
We lower and upper bound the LHS and RHS by considering an alternative random variable $X'$, defined as follows:
Define $p := \Pr[X \le 1]$; note that $0 < p < 1$: the lower bound holds since $\E[X]\leq 1$ and the upper bound holds WLOG, since otherwise $\E[\min(1,X)]=\E[X]$. Define the random variables $Y := [X \mid X \le 1]$ and $Z := [X \mid X > 1]$.\footnote{If $\Pr[X > 1] = 0$, the inequality clearly holds as $\E[\min(1,X)] = \E[X].$} Finally, let the random variable $X'$ take value $\E[Y]$ with probability $p$ and $\E[Z]$ with probability $1 - p$. Note first that $X$ and $X'$ share the same expectation and truncated expectation:
\begin{align*}
& \E[X] = p \cdot \E[X \mid X \le 1] + (1-p) \cdot \E[X \mid X > 1] =  \E[X']\\
& \E[\min(1, X)] = p \cdot \E[X \mid X \le 1] + (1-p) \cdot 1 =  \E[\min(1,X')].
\end{align*}
Additionally, $X'$ is less variable than $X$:
\begin{align*}
\Var(X) &= p \cdot \E[(X-\E[X])^2 \mid X \le 1] + (1- p) \cdot \E[(X - \E[X])^2 \mid X > 1] \\
&= p \cdot \E[ (Y - \E[X])^2 ] + (1- p) \cdot \E[ (Z - \E[X])^2] \\
&\ge p \cdot (\E[Y] - \E[X])^2 + (1- p) \cdot (\E[Z] - E[X])^2 &&\text{Jensen's inequality} \\
&= \Var(X').
\end{align*}
It therefore suffices to prove our desired inequality for the two-point distribution $X'$, which by the above would imply the same inequality for $X$, as follows.
$$\E[\min(1,X)] = \E[\min(1,X')]\geq 
\E[X'] - \frac{1}{2}\cdot \Std(X')\cdot \sqrt{\E[X']} \geq \E[X] - \frac{1}{2}\cdot \Std(X)\cdot \sqrt{\E[X]}.$$

We now prove the desired inequality for $X'$. For ease of notation, define $a := \E[Y]$ and $b := \E[Z]$. 
The desired inequality is therefore as follows:
\begin{align*}
    ap + (1-p) & \ge ap + (1-p)b - \frac{1}{2} \cdot \sqrt{ a^2p + b^2(1-p) - (ap + b(1-p))^2 } \cdot \sqrt{ap + b(1-p)},
\end{align*}
which after some simplification is equivalent to
$$(1-p)(1-b) + \frac{1}{2} \cdot (b-a) \cdot \sqrt{p-p^2} \cdot \sqrt{ap + b(1-p)} \ge 0.$$ 
Recalling that $p<1$ and dividing by $1-p$ results in the equivalent inequality 
$$
1-b + \frac{1}{2} \cdot (b-a) \cdot \sqrt{p} \cdot \sqrt{a\cdot\frac{p}{1-p} + b} \ge 0.$$
We note that the LHS is increasing in $p$, since $b-a \ge 0$ and both $\sqrt{p}$ and $\frac{p}{1-p}$ are increasing for $p \in [0,1)$. Moreover, since $\E[X'] = ap + b(1-p) \le 1$, we have that $p \ge \frac{1-b}{a-b}$. Thus, it suffices to confirm the inequality for $p=(1-b)/(a-b)$, for which this inequality becomes
$$1-b + \frac{1}{2}\cdot (b-a) \sqrt{\frac{1-b}{a-1}}\geq 0.$$ 
For $x := 1 -a \ge 0$ and $y:= b - 1 \ge 0$, dividing through by $\sqrt{\frac{1-b}{a-1}}$ and re-arranging, this is equivalent to $$\frac{1}{2}\cdot(x+y) \ge \sqrt{xy},$$ which holds by the AM-GM inequality. The lemma follows. 
\end{proof}

\subsection{Incomparability of bounds}\label{sec:incomparable}

Our main lower bounds on $\E[\min(1,X)]$ are given by Lemmas \ref{lem:fracbucketingbound} and \ref{lem:Emin1X-variance}, restated below for ease of reference.
We briefly provide examples showing that these two are incomparable, with either one dominating the other for some instances.

\fracbucketing*
\varbound*

\begin{example}[Variance may dominate bucketing] 
Consider an instance with $n=1/\epsilon$ and $c_i=2\epsilon$ and $q_i=1/2$ for all $i\in [n]$. 
For $\theta=1/2$, the fractional bucketing bound \Cref{lem:fracbucketingbound} asserts that  $\E[\min(1,X)]\geq 1-(1/2)^2 = 3/4$. (By inspection, one sees that this choice of $\theta$ maximizes this bound for the instance.)
On the other hand, for this example, by sub-additivity of variance of NCD variables we have that $\Var(X)\leq \sum_i  \Var(c_i \cdot X_i)\leq c_i^2 \cdot q_i = 2\eps.$
Consequently, the variance bound for this instances asserts that 
$\E[\min(1,X)] \geq 1-\sqrt{2\epsilon}.$
\end{example}
\begin{example}[Bucketing may dominate variance]\label{ex:bucketing}
If $X = \sum_{i=1}^n X_i$ for independent $X_i \sim \text{Ber}\left( \frac{1}{n} \right)$, note that as $n \rightarrow \infty$ we have $\E[\min(1,X)] \rightarrow 1-1/e$. The fractional bucketing bound is tight when taking $\theta \in [0, 1-1/n]$. However, $\Var(X) \rightarrow 1$, so the variance bound is loose, giving only a lower bound $\rightarrow 1/2$. 
\end{example}

%% file: app-edgeweighted.tex
\section{Deferred Proofs of Section~\ref{sec:edge-weighted}: Edge-Weighted Matching}\label{app:edgeweighted}

\begin{restatable}{claim}{gdecreasing}\label{claim:gdecreasing}
    For any fixed value of $\theta \in (0,1)$, the function $$g_{\theta}(\mu_L) := \left( 1 - (1-\theta) \cdot \Big\{ \frac{\mu_L}{1-\theta} \Big\} \right) \cdot \theta^{\lfloor \frac{\mu_L}{1-\theta} \rfloor}$$ is a continuous, monotone decreasing function of $\mu_L \in (0,1)$. 
\end{restatable}

\begin{proof}
    For any $\mu_L$ that is not a multiple of $(1-\theta)$, continuity is clear; to check continuity at $k(1-\theta)$ for a positive integer $k$, observe that 
    $$\lim_{\mu_L \rightarrow k(1-\theta)^{-}} g_{\theta}(\mu_L) = \lim_{\mu_L \rightarrow k(1-\theta)^{-}} \left( 1 - (1-\theta) \cdot \left( \frac{\mu_L}{1-\theta} - (k-1) \right) \right) \cdot \theta^{k-1} = \theta^k,$$ while additionally we clearly have
    $$g_{\theta}(k(1-\theta)) = \lim_{\mu_L \rightarrow k(1-\theta)^+} g_{\theta}(\mu_L)= \theta^{k}.$$ 
    Given continuity, to show that $g_{\theta}(\cdot)$ is monotone decreasing it suffices to show that it is decreasing on each interval $((k-1)(1-\theta), k(1-\theta)]$ for positive integers $k$; this can be observed by inspection. 
\end{proof}

\begin{claim}\label{scalingsumto1} For constants $A \in [0,1]$, $B > 0$, and $C \ge 0$, let 
$$ f(x) := \frac{1 - g_\theta \left( C \cdot x\right) \cdot A^{x} }{ B \cdot x}.$$ Then, $f(\cdot)$ is a non-increasing function for $x > 0$.  
\end{claim}

\begin{proof}

By \Cref{claim:gdecreasing}, we have that $f(\cdot)$ is continuous; furthermore we can easily observe that $f(\cdot)$ is differentiable for all $x$ that are not integer multiples of $\frac{1-\theta}{C}$.\footnote{We assume $C >0$ as the lemma is clear if $C=0$.} Thus it suffices to show $f'(x) \le 0$ for all $x > 0$ where this derivative is defined. Indeed, this is equivalent to observing that $$(B \cdot x) \left( - g_{\theta}(C \cdot x) \cdot A^x \cdot \ln(A) - \frac{d}{dx} \Big[ g_{\theta}(C \cdot x) \Big] \cdot A^x \right) - \left( 1 - g_{\theta}(C \cdot x) \cdot A^x \right) \cdot B \le 0,$$ which, after division by $-B \cdot A^x < 0$,  simplifies to $$ x  \cdot  g_{\theta}(C \cdot x)  \cdot \ln(A) + x \cdot  \frac{d}{dx} \Big[ g_{\theta}(C \cdot x) \Big] + \left( A^{-x} - g_{\theta}(C \cdot x)  \right)  \ge 0.$$ We next compute 
\begin{align*}
     \frac{d}{dx} \Big[ g_{\theta}(C \cdot x) \Big] &= \frac{d}{dx} \left[ \left( 1 - (1-\theta) \cdot \Big\{ \frac{C \cdot x}{1-\theta} \Big\} \right) \cdot \theta^{\lfloor \frac{C \cdot x}{1-\theta} \rfloor} \right] =   (-C) \cdot \theta^{\lfloor \frac{C \cdot x}{1-\theta} \rfloor}  
\end{align*}
(where this derivative is defined). Substituting, it suffices to show $$ x  \cdot  g_{\theta}(C \cdot x)  \cdot \ln(A) - C \cdot x \cdot \theta^{\lfloor \frac{C \cdot x}{1-\theta} \rfloor} + \left( A^{-x} - g_{\theta}(C \cdot x)  \right)  \ge 0.$$ 

We claim the LHS is decreasing in $A$ when holding all other parameters constant. Indeed, the derivative with respect to $A$ is given by $$ x \cdot g_{\theta}(C \cdot x) \cdot \frac{1}{A} - x \cdot A^{-x-1} \le 0,$$ where the inequality follows by noting $g_{\theta}(C \cdot x) \le 1 \le A^{-x}.$

    Hence it suffices to prove the desired inequality for $A = 1$, i.e., $$g_{\theta}(C \cdot x) \le 1 - C \cdot x \cdot  \theta^{\lfloor \frac{C \cdot x}{1-\theta} \rfloor}.$$
By the definition of $g_{\theta}( \cdot )$ this inequality is equivalent to $$\left( 1 + C \cdot x - (1-\theta) \cdot \left\{ \frac{C \cdot x}{1 - \theta} \right\} \right) \cdot \theta^{\lfloor \frac{C \cdot x}{1-\theta} \rfloor} \le 1.$$ After substituting $\left\{ \frac{C \cdot x}{1 - \theta} \right\} = \frac{C \cdot x}{1-\theta} - \lfloor \frac{C \cdot x}{1-\theta} \rfloor$, this reduces to showing $\left( 1 + (1-\theta) \cdot k \right) \cdot \theta^k \le 1$ for non-negative integers $k$, which is clear for $\theta \in [0,1]$. 
\end{proof}

\begin{restatable}{claim}{lipschitz}\label{claim:lipschitz}
    Define $$B(x) := 1 - \exp(-(1+\delta)(1-x)) \cdot g_{\hat{\theta}}((1-\eps) x).$$ Then, $B(\cdot)$ is 3-Lipschitz for $x \ge 0$, i.e., for any $x, y \ge 0$ we have $|B(x) - B(y)| \le 3 \cdot|x-y|$.
\end{restatable}

\begin{proof}
    By \Cref{claim:gdecreasing}, we have that $B(x)$ is a continuous function. It is also straightforward to see that for fixed $\eps, \delta > 0$ we have that $B(\cdot)$ is differentiable on $[0,1]$ for all but finitely many points $x$, when $(1-\eps)x$ is an integer multiple of $1-\hat{\theta}$. Thus, it suffices to show that $|B'(x)| \le 4$ for all $x \in [0,1]$ where $B(\cdot)$ is differentiable. 

    For convenience of notation, let $a(x) := \exp(-(1+\delta) x)$ and $b(x) := g_{\hat{\theta}}((1-\eps) x)$. Observe that for $x \in [0,1]$, we have $0 \le a(1-x), b(x) \le 1$ and $|a'(x)| \le (1+\delta) \le 2$. Additionally, as in \Cref{scalingsumto1}, we have $|b'(x)| \le 1$. Thus
    \begin{align*}
        |B'(x)| = |a(x)| \cdot |b'(x)| + |a'(x)| \cdot |b(x)| \le 3
    \end{align*}we have that $B(x) = 1 - f(x) \cdot b(x)$ is 3-Lipschitz.
\end{proof}

%% file: app-vertexweighted.tex
\section{Deferred Proofs of Section~\ref{sec:vertexweighted}: Vertex-Weighted Matching } \label{app:vertexweighted}

\betaHbound*
\begin{proof}
It will be convenient to define the contribution of high-degree edges to $\alpha$ by 
\begin{align*}
    \alpha^{> \theta} & := \frac{\sum_{i,t} r_{i,t } \cdot y_{i,t} \cdot x_{i,t} \cdot \mathbbm{1}[y_{i,t} 
 > \theta]}{\sum_{i,t} x_{i,t}}. \end{align*}

 Clearly it suffices to show $\alpha^{> \theta} \le S^{> \theta} - \frac{(S^{> \theta})^2}{2}$. We will show this for every offline node $i$; by Jensen's inequality this implies the global result. Formally, for every offline node $i$ define 
 \begin{align*}\alpha^{> \theta}_i & := \frac{\sum_t r_{i,t} \cdot x_{i,t}  \cdot y_{i,t} \cdot \mathbbm{1}[y_{i,t} > \theta]}{\sum_t x_{i,t}} \\ S^{> \theta}_i &:= \frac{\sum_t r_{i,t} \cdot x_{i,t} \cdot \mathbbm{1}[y_{i,t} > \theta]}{\sum_t x_{i,t}}.
 \end{align*}
 Suppose $\alpha^{> \theta}_i \le f(S_i^{> \theta})$ for $f(z) := z - \frac{z^2}{2}$. Then, by Jensen's inequality and the concavity of $f(\cdot)$,
$$\alpha^{> \theta} = \sum_i \alpha_i^{> \theta} \cdot \frac{ \sum_t x_{i,t} }{\sum_{i,t} x_{i,t}} \le \sum_i f(S^{> \theta}_i)  \cdot \frac{ \sum_t x_{i,t} }{\sum_{i,t} x_{i,t}} \le f \left( \sum_i S_i^{> \theta} \cdot  \frac{ \sum_t x_{i,t} }{\sum_{i,t} x_{i,t}}  \right) =S ^{>\theta} - \frac{(S^{> \theta})^2}{2}.$$ Thus it remains to show  $\alpha_i^{> \theta} \leq S_i^{> \theta} - \left( S_i^{> \theta} \right)^2/2$ for any fixed offline node $i$. For convenience of notation, after fixing $i$ we let $H$ denote the set of all online arrivals $t$ with $y_{i,t} > \theta$, and drop the ``$\mathbbm{1}[y_{i,t}> \theta$]'' indicators. Note that by the fractional degree constraint  (\Cref{lem:degree-implicit}), we have
\begin{align}\label{eqn:alpha_i-bound}
\alpha_i^{> \theta} & = \frac{\sum_{t \in H} x_{i,t} \cdot r_{i,t} \cdot y_{i,t}}{\sum_{t} x_{i,t}} \le \frac{\sum_{t  \in H} x_{i,t} \cdot r_{i,t} \cdot y_{i,t}}{ \left( \sum_t x_{i,t} \right)^2}.
\end{align}
We therefore turn to upper bounding the RHS of \Cref{eqn:alpha_i-bound}. Equivalently (and to avoid notational clutter), we 
upper bound $(\star) := \sum_{t \in H} x_{i,t} \cdot r_{i,t} \cdot y_{i,t}$ as follows:
\begin{align}
(\star) &= \sum_{t \in H} x_{i,t} \cdot r_{i,t} \cdot \left( \sum_{t'} x_{i,t'} - \sum_{t' \ge t} x_{i,t'} \right)  \nonumber \\
&= \left(\sum_{t \in H} x_{i,t} \cdot r_{i,t} \right) \left( \sum_{t} x_{i,t} \right) - \sum_{t \in H} x_{i,t}^2 \cdot r_{i,t} - \sum_{t \in H, t' \in H, t' > t} x_{i,t} \cdot r_{i,t} \cdot x_{i,t'} \nonumber \\
&\le \left(\sum_{t \in H} x_{i,t} \cdot r_{i,t} \right) \left( \sum_{t} x_{i,t} \right)  - \sum_{t \in H} x_{i,t}^2 \cdot r_{i,t}^2 - \sum_{t \in H, t' \in H, t' > t} x_{i,t} \cdot r_{i,t} \cdot x_{i,t'} \cdot r_{i,t'} \label{eqn:6.5firstineq} \\
&\le \left(\sum_{t \in H} x_{i,t} \cdot r_{i,t} \right) \left( \sum_{t} x_{i,t} \right)  - \frac{1}{2} \cdot \sum_{t \in H} x_{i,t}^2 \cdot r_{i,t}^2 - \sum_{t \in H, t' \in H, t' > t} x_{i,t} \cdot r_{i,t} \cdot x_{i,t'} \cdot r_{i,t'}  \label{eqn:6.5secondineq} \\
&=  \left(\sum_{t \in H} x_{i,t} \cdot r_{i,t} \right) \left( \sum_{t} x_{i,t} \right) - \frac{1}{2} \cdot \left( \sum_{t \in H} x_{i,t} \cdot r_{i,t} \right)^2. \nonumber
\end{align}
Here, \eqref{eqn:6.5firstineq} follows by noting $r_{i,t}, r_{i,t'} \in [0,1]$ and $x_{i,t}, x_{i,t'} \ge 0$, and \eqref{eqn:6.5secondineq} simply notes that $x_{i,t}^2 \cdot r_{i,t}^2 \ge 0$. 
Dividing both sides by $\left( \sum_t x_{i,t} \right)^2$ and combining with \eqref{eqn:alpha_i-bound}, we obtain the desired inequality, 
\begin{align*}
\alpha_i^{> \theta} & \leq \frac{ \sum_{t \in H} x_{i,t} \cdot r_{i,t} \cdot y_{i,t}} {\left(\sum_t x_{i,t}\right)^2 } = \frac{(\star)} {\left(\sum_t x_{i,t}\right)^2 }\le S_i^{> \theta} - \frac{(S_i^{> \theta})^2}{2}. \qedhere 
\end{align*}
\end{proof}

\betaLbound*
 \begin{proof}
     Similar to the previous proof, we define 
\begin{align*}
    \alpha^{\le \theta} & := \frac{\sum_{i,t} r_{i,t } \cdot y_{i,t} \cdot x_{i,t} \cdot \mathbbm{1}[y_{i,t} 
 \le \theta]}{\sum_{i,t} x_{i,t}} \\
 S_i^{\le \theta} & := \frac{\sum_{t} r_{i,t } \cdot x_{i,t} \cdot \mathbbm{1}[y_{i,t} 
 \le \theta]}{\sum_{t} x_{i,t}}.\end{align*} By the same logic it suffices to show $\alpha_i^{\le \theta} \le S_i^{\le \theta} \cdot \left( \theta - \frac{S^{\le \theta}}{2} \right)$. For convenience of notation, we again fix $i$ and let $L$ denote the online arrivals $t$ with $y_{i,t} \le \theta$, allowing us to drop the ``$\mathbbm{1}[y_{i,t} \le \theta]$'' indicators. 

Again by the fractional degree constraint (\Cref{lem:degree-implicit}), we have
\begin{align}\label{eqn:alpha_i-bound-below}
\alpha_i^{\le \theta} & = \frac{\sum_{t \in L} x_{i,t} \cdot r_{i,t} \cdot y_{i,t}}{\sum_t x_{i,t}} \le \frac{\sum_{t \in L} x_{i,t} \cdot r_{i,t} \cdot y_{i,t}}{ \left( \sum_t x_{i,t} \right)^2}.
\end{align}
As before we upper bound $(\star) := \sum_t x_{i,t} \cdot r_{i,t} \cdot y_{i,t}$. The proof precedes similarly to that of \Cref{lem:betaHbound}, with the additional fact that $\sum_{t' \notin L} x_{i, t'} \le 1 - \theta$ by the definition of $L$. 
\begin{align}
(\star) &= \sum_{t \in L} x_{i,t} \cdot r_{i,t} \cdot \left( \sum_{t'} x_{i,t'} - \sum_{t' \ge t} x_{i,t'} \right) \nonumber \\
&= \left(\sum_{t \in L} x_{i,t} \cdot r_{i,t} \right) \left( \sum_{t} x_{i,t} \right) - \sum_{t \in L} x_{i,t}^2 \cdot r_{i,t} - \sum_{t \in L, t' > t} x_{i,t} \cdot r_{i,t} \cdot x_{i,t'} \nonumber \\
&= \left(\sum_{t \in L} x_{i,t} \cdot r_{i,t} \right) \left( \sum_{t} x_{i,t} \right) - \sum_{t \in L} x_{i,t}^2 \cdot r_{i,t} - \sum_{t \in L, t' \in L, t' > t} x_{i,t} \cdot r_{i,t} \cdot x_{i,t'} - \sum_{t \in L, t' \notin L } x_{i,t} \cdot r_{i,t} \cdot x_{i,t'} \nonumber \\
&\le \left(\sum_{t \in L} x_{i,t} \cdot r_{i,t} \right) \left( \sum_{t} x_{i,t} \right) - \sum_{t \in L} x_{i,t}^2 \cdot r_{i,t} - \sum_{t \in L, t' \in L, t' > t} x_{i,t} \cdot r_{i,t} \cdot x_{i,t'} - (1 - \theta) \cdot \sum_{t \in L  } x_{i,t} \cdot r_{i,t} \label{6.6firstineq}   \\
&= \left(\sum_{t \in L} x_{i,t} \cdot r_{i,t} \right) \left( \sum_{t} x_{i,t} -1 + \theta \right) - \sum_{t \in L} x_{i,t}^2 \cdot r_{i,t} - \sum_{t \in L, t' \in L, t' > t} x_{i,t} \cdot r_{i,t} \cdot x_{i,t'} \nonumber   \\
&\le \left(\sum_{t \in L} x_{i,t} \cdot r_{i,t} \right) \left( \sum_{t} x_{i,t} \right) \cdot \theta - \sum_{t \in L} x_{i,t}^2 \cdot r_{i,t} - \sum_{t \in L, t' \in L, t' > t} x_{i,t} \cdot r_{i,t} \cdot x_{i,t'}  \label{6.6secondineq}  \\
&\le \left(\sum_{t \in L} x_{i,t} \cdot r_{i,t} \right) \left( \sum_{t} x_{i,t} \right) \cdot \theta  - \sum_{t \in L} x_{i,t}^2 \cdot r_{i,t}^2 - \sum_{t \in L, t' \in L, t' > t} x_{i,t} \cdot r_{i,t} \cdot x_{i,t'} \cdot r_{i,t'} \label{6.6thirdineq}   \\
&\le \left(\sum_{t \in L} x_{i,t} \cdot r_{i,t} \right) \left( \sum_{t} x_{i,t} \right) \cdot \theta  - \frac{1}{2} \cdot \sum_{t \in L} x_{i,t}^2 \cdot r_{i,t}^2 - \sum_{t \in L, t' \in L, t' > t} x_{i,t} \cdot r_{i,t} \cdot x_{i,t'} \cdot r_{i,t'} \label{6.6fourthineq}  \\
&=  \left(\sum_{t \in L} x_{i,t} \cdot r_{i,t} \right) \left( \sum_{t} x_{i,t} \right) \cdot \theta  - \frac{1}{2} \cdot \left( \sum_{t \in L} x_{i,t} \cdot r_{i,t} \right)^2. \nonumber
\end{align}
Above, \eqref{6.6firstineq} relied on $\sum_{t' \notin L} x_{i,t'} \le 1 - \theta$, \eqref{6.6secondineq} relied on $\sum_t x_{i,t} \le 1$. As in the proof of \Cref{lem:betaHbound}, \eqref{6.6thirdineq} and \eqref{6.6fourthineq} follow from observing $x_{i,t}, x_{i,t'}, r_{i,t}, r_{i,t'} \in [0,1]$.

Dividing both sides by $\left( \sum_t x_{i,t} \right)^2$ and combining with \eqref{eqn:alpha_i-bound-below}, we obtain the desired inequality, 
\begin{align*}
\alpha_i^{\le \theta} & \leq \frac{ \sum_{t \in L} x_{i,t} \cdot r_{i,t} \cdot y_{i,t}} {\left(\sum_t x_{i,t}\right)^2 } = \frac{(\star)} {\left(\sum_t x_{i,t}\right)^2 }\le S_i^{\le \theta} \cdot \theta - \frac{(S_i^{\le \theta})^2}{2}. \qedhere 
\end{align*}
\end{proof}

\begin{restatable}{lemma}{linearfunctionlowerbound}\label{lem:linearfunctionlowerbound} For $0 \le x < 1$ and $0 < y \le (1-x)^2$ we have
     $$ 0.613 + 0.122x + 0.21y \le 1 - g_{\nicefrac{1}{2}}  \left(  x\right) \cdot  \left( 1 - \frac{y}{1 - x } \right)^{\frac{(1- x)^2}{y}}.$$
\end{restatable}

\begin{proof}
We first show the function $$f(x,y) := \left( 1 - \frac{y}{1-x} \right)^{\frac{(1-x)^2}{y}}$$ is 1-Lipschitz with respect to the $\ell_1$-norm on the domain $\{0 \le x < 1, 0 < y \le (1-x)^2\}$.\footnote{I.e., for $(x_1, y_1)$ and $(x_2, y_2)$ in our domain, we have $\left| \frac{y_1}{1-x_1} - \frac{y_2}{1-x_2} \right| \le \| (x_1,y_1) - (x_2,y_2) \|_1 =  |x_1 - x_2| + |y_1 - y_2|$.}

For this it suffices to show that $\left| \frac{\partial}{\partial x} f(x,y) \right|$ and $\left| \frac{\partial}{\partial y} f(x,y) \right|$ are both bounded by 1 on our domain. We start by bounding the partial derivative with respect to $x$.

For ease of notation, it will be convenient to define the function $r_C(z) := \left( 1 - z \right)^{C / z^2}$ for a constant $C \in [0,1]$ and bound its derivative, as $f(x,y) = r_{y} \left( \frac{y}{1-x} \right)$.

We claim $\frac{r_C(z)}{ 1-z }$ is non-decreasing on its domain. Indeed, differentiating this is equivalent to $$\frac{(1-z) \cdot r_C'(z) + r_C(z)}{(1-z)^2} \ge 0 \quad \Leftrightarrow \quad r_C'(z) \ge -\frac{r_C(z)}{1-z}.$$

We compute 
\begin{align*}
    r_C'(z) = C \cdot (1-z)^{C/z^2 - 1} \cdot \frac{2(z-1) \cdot \log(1-z) - z}{z^3}.
\end{align*}

The following fact will be useful. 
\begin{fact}\label{fact:appendixlogbounding}
    For $x \in [0,1)$, $x \cdot (1-x) \le (x-1) \cdot \log(1-x) \le x.$
\end{fact}
The lower bound is classic; the upper bound is follows by noting that $x + (1-x) \log(1-x) \ge 0$ as the derivative of the LHS is $- \log(1-x) \ge 0$, and evaluating the LHS at $x=0$ produces $0$. This fact allows us to bound 
$$C \cdot (1-z)^{C/z^2 - 1} \cdot \frac{1-2z}{z^2} \le r_C'(z) \le C \cdot (1-z)^{C/z^2 - 1} \cdot \frac{1}{z^2}$$ which can be equivalently written as 
$$C \cdot \frac{r_C(z)}{1-z} \cdot \frac{1-2z}{z^2} \le r_C'(z) \le C \cdot \frac{r_C(z)}{1-z} \cdot \frac{1}{z^2}$$

Looking closely at the lower bound we see $$r_{C}'(z) \ge  C \cdot \frac{r_C(z)}{1-z} \cdot \frac{1-2z}{z^2} \ge - \frac{r_C(z)}{1-z} $$ which suffices to argue that $\frac{r_C(z)}{1-z}$ is non-decreasing on its domain. Note then that $$| r_C'(z) | \le C \cdot \frac{r_C(z)}{1-z} \cdot \frac{1}{z^2} \le C \cdot \frac{r_C(z)}{1-z} \cdot \frac{ \max(1, |1-2z|)}{z^2} \le C \cdot \frac{r_C(z)}{1-z} \cdot \frac{1}{z^2}.$$

We are now ready to bound our original partial derivative as follows:
\begin{align*}
    \left| \frac{\partial}{\partial x} f(x,y) \right| &= \left| \frac{d}{dx} r_y \left( \frac{y}{1-x} \right) \right| \\
    &=  \left| \frac{y}{(1-x)^2} \cdot r_y' \left( \frac{y}{1-x} \right) \right| \\
    &\le \frac{y}{(1-x)^2} \cdot y \cdot \frac{r_y \left( \frac{y}{1-x} \right) }{1 - \frac{y}{1-x}} \cdot \frac{(1-x)^2}{y^2} \\
    &=  \frac{r_y \left( \frac{y}{1-x} \right) }{1 - \frac{y}{1-x}} \\
    &\le  \frac{r_y \left( \frac{y}{\sqrt{y}} \right) }{1 - \frac{y}{\sqrt{y}}} && \frac{r_C(z)}{1-z} \text{ non-decreasing, } x \le 1 - \sqrt{y} \\
    &= \frac{(1-\sqrt{y})^{\frac{y}{(\sqrt{y})^2}}}{1-\sqrt{y}} = 1.
\end{align*}

We now turn to bounding the partial derivative with respect to $y$, which proceeds similarly (and is actually slightly simpler). For any constant $C \in [0,1]$ we define $s_{C}(z) := \left( 1 - z \right)^{C/z} = \exp \left( \frac{C}{z} \cdot \ln (1-z) \right)$. We first show $\frac{s_C(z)}{1-z}$ is non-decreasing; indeed, after differentiating this reduces to (as before)  $$s_C'(z) \ge -\frac{s_C(z)}{1-z}.$$ We then compute
\begin{align*}
    s_C'(z) = C \cdot (1-z)^{\frac{C}{z} - 1} \cdot \frac{(z-1) \cdot \ln(1-z) -z}{z^2} = C \cdot \frac{s_C(z)}{1-z} \cdot  \frac{(z-1) \cdot \ln(1-z) -z}{z^2}.
\end{align*}
Observe $ \frac{(z-1) \cdot \ln(1-z) -z}{z^2} \in [-1, 0]$ by \Cref{fact:appendixlogbounding}. This suffices to note that $$0 \ge s_C'(z) \ge -C \cdot \frac{s_C(z)}{1-z} \ge - \frac{s_C(z)}{1-z}.$$ So, $\frac{s_C(z)}{1-z}$ is non-decreasing. Now observe 
\begin{align*}
    \left| \frac{\partial}{\partial y} f(x,y) \right| &= \left| \frac{d}{dy} s_{1-x} \left( \frac{y}{1-x} \right) \right| \\
    &\le \frac{1}{1-x} \cdot \left| s_{1-x}' \left( \frac{y}{1-x} \right) \right| \\
    &\le \frac{1}{1-x} \cdot (1-x) \cdot \frac{s_{1-x} \left(  \frac{y}{1-x} \right) }{1- \frac{y}{1-x}} \\
    &\le \frac{s_{1-x} \left( 1-x \right) }{1-(1-x)} &&  \frac{s_C(z)}{1-z} \text{ non-decreasing}, \frac{y}{1-x} \le 1-x  \\
    &= 1.
\end{align*}

This implies that $f(x,y)$ is 1-Lipschitz in its arguments (with respect to the $\ell_1$-norm), and the rest of the proof is straightforward. As in the proof of \Cref{scalingsumto1} we can compute $g_{\nicefrac{1}{2}}'(x) = - \left( \frac{1}{2} \right)^{\lfloor 2x \rfloor}$ for $x \in [0,1] \setminus \{0, \nicefrac{1}{2}, 1\}.$ As $g_{\nicefrac{1}{2}}(x)$ is continuous, we can observe from this that it is $1$-Lipschitz. Note additionally that $g_{\nicefrac{1}{2}}(x)$ takes values in $[0,1]$ for $x \in [0,1]$, and $f(x,y)$ takes values in $[0,1]$ for $x, y$ in its domain $\{ 0 \le x \le 1, 0 < y \le (1-x)^2 \}.$ This implies that the product $g_{\nicefrac{1}{2}}(x) \cdot f(x,y)$ is 2-Lipschitz. Thus we can (loosely) bound that $$1 - g_{\nicefrac{1}{2}}(x) \cdot \left( 1 - \frac{y}{1-x} \right)^{\frac{(1-x)^2}{y}} - \acoeff - \bcoeff \cdot x - \ccoeff \cdot y $$ is $3$-Lipschitz. We would like to show that the above expression is non-negative; by our Lipschitz condition it suffices to evaluate it at a grid with spacing $10^{-4}$ (for both $x$ and $y$) and show it is always at least $3 \cdot 10^{-4}$. This is straightforward to check computationally.\footnote{See code at \url{https://tinyurl.com/mts6u5ua}.}
\end{proof}

\begin{restatable}{lemma}{finalvertexweightedbound}\label{lem:finalvertexweightedbound}
    Let $$B(x,y) :=  \max \left( 1 - \frac{1}{2} \sqrt{ x +y - x\cdot \left( \frac{1}{2} + \frac{x}{2} \right) -  \frac{y^2}{2} },  \acoeff + \bcoeff \cdot 0.5 + \ccoeff \cdot \frac{y^2}{2}  \right).$$ Then $(\dagger)_{\nicefrac{1}{2}} \ge \min_{x, y \in [0, 0.5]^2} B(x,y) \ge \vertexapprox.$ 
\end{restatable}
\begin{proof}
    It suffices to prove a Lipschitzness condition for $B(x,y)$ and search over a suitably fine-grid, but care must be taken because the function $\sqrt{z}$ is not Lipschitz as $z \rightarrow 0$. However, this is easy to get around by restricting our domain; in particular we claim that it suffices to search over $y \in [0.25, 0.5]$. To see why, note that $x - x \cdot \left( \frac{1}{2} + \frac{x}{2} \right) \in [0, 0.125]$ for $x \in [0, 0.5]$. So, if $y \le 0.25$, we have $$ B(x,y) \ge 1 - \frac{1}{2} \sqrt{ x +y - x\cdot \left( \frac{1}{2} + \frac{x}{2} \right) -  \frac{y^2}{2} } \ge 1 - \frac{1}{2} \sqrt{ 0.125 + y - \frac{y^2}{2} } \ge 0.7.$$

    Hence we restrict our domain to $x \in [0, 0.5]$ and $y \in [0.25, 0.5]$, on which we claim our function is 1-Lipschitz (with respect to the $\ell_1$-norm). Here, we can first loosely observe that $ \acoeff + \bcoeff \cdot 0.5 + \ccoeff \cdot \frac{y^2}{2}  $ is clearly $1$-Lipschitz for $y \in [0.25, 0.5]$. We next observe that $x +y - x\cdot \left( \frac{1}{2} + \frac{x}{2} \right) -  \frac{y^2}{2}$ is 1-Lipschitz in $x$ when $y$ is fixed and 1-Lipschitz in $y$ when $x$ is fixed. Hence it is 1-Lipschitz in $(x,y)$ with respect to the $l_1$-norm. The function $1 - \frac{1}{2}\sqrt{z}$ is 1-Lipschitz in $z$ as long as $z \ge 0.1$; by our condition that $y \ge 0.25$ this is sufficient to see that $1 - \frac{1}{2} \sqrt{x +y - x\cdot \left( \frac{1}{2} + \frac{x}{2} \right) -  \frac{y^2}{2}}$ is 1-Lipschitz. Finally, $\max(x,y)$ is 1-Lipschitz, which implies our claim. 

    Thus minimizing $B(x,y)$ on a grid with step size $10^{-4}$ has error at most $10^{-4}$, for $y \in [0.25, 0.5]$. Linked code demonstrates that the minimum on such a grid is at least $0.685 + 10^{-4}$.\footnote{See code at \url{https://tinyurl.com/467f3zde}.} 
\end{proof}

%% file: app-hardness.tex
\section{Deferred Proofs of Section~\ref{sec:hardness}: Hardness}\label{app:hardness}
We first formalize the bijection between algorithms for the $\textsc{Stochastic-3-Sat}$ instance $\phi$, and matching algorithms for $\mathcal{I}_{\phi}$ which match every arriving variable node. In particular, setting an odd-indexed variable $x_{2i+1}$ to \texttt{True} corresponds to matching the $(2i+1)$\textsuperscript{th} arrival to $T^i$, while setting $x_{2i+1}$ to \texttt{False} corresponds to matching the $(2i+1)$\textsuperscript{th} arrival to $F^i$. Nature setting  $x_{2i}$ to \texttt{True} corresponds to the $2i$\textsuperscript{th} online node not arriving, while nature setting $x_{2i}$ to \texttt{False} corresponds to matching the $(2i)$\textsuperscript{th} arrival to $F^i$. In this way, if a matching algorithm for $\mathcal{I}_{\phi}$ matches all arriving nodes, we may refer to it as ``satisfying'' certain clauses. 

The following observation holds by a simple exchange argument. 

\begin{obs}\label{obs:matchallvarnodespspace}
    The optimal online algorithm for $\calI_{\phi}$ matches all arriving variable nodes.  
\end{obs}

The following observations follows by noting that we use that no variable $x_{2i}$ appears negated in any clause of $\phi$.

\begin{obs}\label{obs:matchcppspace}
Say an online algorithm for $\mathcal{I}_{\phi}$ matches all of the first $n$ variable nodes that arrive. Fix a clause $C \in \mathcal{C}$. \begin{enumerate}[label=(\alph*)]
\item If $C$ was not satisfied, we cannot match the corresponding clause node.  
\item If $C$ was satisfied, we can guarantee matching the corresponding clause node with probability $p$. (For example, we can discard all arriving stochastic nodes for clauses other than $C$.)

\end{enumerate}
\end{obs}

We emphasize that this claim only holds for individual clauses $C$ and not all clauses simultaneously. We now bound the performance of the optimal online algorithm on $\calI_{\phi}$ in terms of $\textsc{Opt}$ to complete the reduction. 

\begin{claim} \label{claim:optonatmostpspace}
Any online algorithm for the matching instance $\calI_{\phi}$ matches at most $ \frac{3n}{4} + p \cdot \textsc{Opt}$ nodes (in expectation). 
\end{claim}
\begin{proof}
    By \Cref{obs:matchallvarnodespspace}, we can restrict our attention to algorithm $\mathcal{A}$ which match all arriving variable nodes. The expected number of variable nodes that arrive is precisely $\frac{3n}{4}$; moreover, if $\mathcal{A}$ satisfies $k$ clauses in expectation, by \Cref{obs:matchcppspace} it matches at most $k \cdot p$ stochastic nodes. 
\end{proof}

We also have the following lower bound.

\begin{restatable}{claim}{optonatleastpspace} \label{claim:optonatleastpspace}
There exists an online algorithm on $\calI_{\phi}$ which matches at least $ \frac{3n}{4} + p \cdot \textsc{Opt} - n \cdot C_k \cdot p^2 $ nodes (in expectation) for a constant $C_k$. 
\end{restatable}

\begin{proof}
    Consider the matching algorithm corresponding to optimum online algorithm for $\phi$; this algorithm matches all arriving variable nodes ($3n/4$ in expectation) and satisfies $\textsc{Opt}$ clauses in expectation. For each $C \in \mathcal{C}$, let $A_C$ denote an edge in $\mathcal{I}_{\phi}$ which would be greedily matched conditioned on no stochastic nodes corresponding to other clauses arriving. Define $A := \cup_{C \in \mathcal{C}} A_{C}$ and note that by \Cref{obs:matchcppspace} we have $\E[|A|] = p \cdot \textsc{Opt}$. As $A$ might not be a matching, we define $B$ to be the set of all arriving edges from clause nodes that are not the first arriving clause-edge to their offline node. Our algorithm's performance is then lower bounded by $$\frac{3n}{4} + \E[|A|] - \E[|B|].$$  
    To upper bound $\E[|B|]$ we will use our assumption that $\phi$ is bounded degree. In particular this implies that a fixed offline node $i$ in $\mathcal{I}_{\phi}$ has at most $k$ edges from clause nodes that may arrive. If $X_i$ denotes the (random) number of edges incident to $i$ from clause nodes that arrive, we note that 
    \begin{align*}
    \E[|B|] &= \sum_{i=1}^{2n} \E[\max(X_i - 1, 0)]  \le 2n \cdot \sum_{i=2}^k i \cdot  \binom{k}{i} \cdot p^i \cdot (1-p)^i 
    \le  (k-1) \cdot k \cdot \binom{k}{k/2} \cdot p^2 
    \le C_k \cdot p^2
    \end{align*}
    where $C_k = k^2 \cdot 2^k $ is a constant.  Hence, the expected gain of our algorithm is at least 
    \begin{align*}
    &\frac{3n}{4} + p \cdot \textsc{Opt} - n \cdot p^2 \cdot C_k . \qedhere 
    \end{align*}
\end{proof}

    To conclude, say we have an algorithm $\mathcal{A}$ for online stochastic matching that is $\beta$-approximate for $\beta < 1$. Note that, by \Cref{claim:optonatmostpspace} and \Cref{claim:optonatleastpspace}, we have that $$\frac{\mathcal{A}(\phi) - 3n/4}{p} \in \left[ \frac{\beta(3n/4 + p \cdot \textsc{Opt} - n \cdot p^2 \cdot C_k) - 3n/4}{p},  \textsc{Opt} \right].$$ The lower bound of this interval simplifies to 
    $\textsc{Opt} \cdot \beta + n \cdot \frac{\left( \frac{3}{4} \left( \beta - 1 \right) - \beta \cdot p^2 \cdot C_k \right)}{p} . $ Note that $\textsc{Opt} \ge \frac{7}{8} \cdot C$ (achieved by setting all variables uniformly at random) and $n \le 3C \le \frac{24}{7} \cdot \textsc{Opt}$. So, the lower bound of the interval is at least $\textsc{Opt} \cdot \beta + \frac{24}{7} \cdot \textsc{Opt} \cdot \frac{\left( \frac{3}{4} \left( \beta - 1 \right) - \beta \cdot p^2 \cdot C_k \right)}{p} . $ Thus, if there exists some $p$ and $\beta$  such that $\beta + \frac{24}{7}   \cdot \frac{\left( \frac{3}{4} \left( \beta - 1 \right) - \beta \cdot p^2 \cdot C_k \right)}{p} \ge \alpha$ then it is $\mathsf{PSPACE}$-hard to approximate the optimal online algorithm for instances of the form $\mathcal{I}_{\phi}$. Take $p = (1 - \sqrt{\alpha}) \cdot \frac{1}{\frac{24}{7} \cdot C_k}$; this is then equivalent to $\beta + K \cdot (\beta - 1) - \beta \cdot (1 - \sqrt{\alpha}) \ge \alpha$ for the universal constant $K = \frac{18}{7p}$. It suffices to take $\beta = \frac{\alpha + K}{\sqrt{\alpha} + K}$ which is some constant smaller than 1.

%% file: non-bernoulli.tex
\section{Generalizing to non-Bernoulli arrivals}\label{app:general}

For simplicity, our previous algorithm and analysis dealt with only the case of ``Bernoulli arrivals,'' assuming that each online node $t$ arrived with some fixed neighborhood with probability $p_t$, and with probability $1-p_t$ didn't neighbor anyone. In this section, we show that our main edge-weighted result (and necessary analysis) extend to the more general non-Bernoulli case. In particular, we now assume that online node $t$ realizes \emph{type} $j$ with probability $p_{j,t}$, where the type $j$ specifies weights $\{w_{i,j,t}\}_i$ to offline nodes. 

As in \cite{papadimitriou2021online, braverman2022max,naor2023online} we can generalize our LP relaxation as follows:

\noindent \textbf{LP Relaxation for General Arrivals.}
\begin{align} \tag{LP-OPTon-Gen} \label{LP-PPSW-Gen}
    \nonumber \max & \sum_{(i,j)\in E} \sum_{t} w_{i, j, t} \cdot x_{i,j,t} \\
    \textrm{s.t.} 
    \enspace 
    &\sum_{i} x_{i, j, t} \le p_{j,t} && \text{for all }j,t \label{eqn:OnlineMatchingConstraintgen}\\
    & 0\leq x_{i, j, t} \le p_{j,t} \cdot \left( 1 - \sum_{t' < t}\sum_{j'} x_{i, j', t'} \right) && \text{for all } i,j,t \label{eqn:PPSWConstraintgen}
 \end{align}

We consider the following algorithm. Here, we do not apply independent discarding, and hence do not obtain negative cylinder dependence of offline nodes. However, we do obtain an upper bound on the probability all nodes are occupied, which suffices to apply the fractional bucketing bound, as we show below. 

\begin{algorithm}[H]
	\caption{Online Correlated Proposals (Generalized)}
	\label{alg:proposals-core-generalized}
	\begin{algorithmic}[1]
 \Statex \textbf{Input:} A vector $\vec x$ satisfying \eqref{LP-PPSW-Gen}~constraint~\eqref{eqn:PPSWConstraintgen} 
 \medskip
\State $\calM\gets \emptyset,\;  F_1\gets [n]$
		\Comment{$\calM$ is the output matching}
		\ForAll{times $t$} \label{line:loop-start-gen}
        \State $F_{t+1}\gets F_t$ \Comment{initially, no free node $i\in F_t$ is matched before time $t+1$}
        \State{$j \gets$ arrival at time $t$}
		\ForAll{offline nodes $i$}
		\State $r_{i,j,t} \gets \frac{x_{i,j,t}}{p_{j,t}\cdot \left(1-\sum_{j, t'<t} x_{i,j,t'}\right)}$
        \EndFor
        \State Let $\vec v_j$ be the vector $(r_{i,j,t})_{i \in F_t}$ indexed by $i$ sorted in decreasing order of $w_{i,j,t}$
        \State Let $I_{j,t} \gets \textsf{PS}(F_t, \vec v_j)$ \label{line:calltoSRgen}
    
        \State Pick some $i^* \in \arg\max_{i\in I_{j,t}} \{w_{i,j,t}\}$  
         \State Add $(i^*, t)$ to the matching $\calM$ and set $F_{t+1}\gets F_{t+1}\setminus \{ i^* \}$  \label{line:matchtgen}
     
        \EndFor
        
	\end{algorithmic}
\end{algorithm}

We first claim that this algorithm is well-defined. Indeed, by Constraint (\ref{eqn:PPSWConstraintgen}) of our generalized LP we have that each $r_{i,j,t} \le 1$, so the call to $\textsf{PS}(\cdot, \cdot)$ in Line~\ref{line:calltoSRgen} is well-defined. Additionally, if in \Cref{line:matchtgen} we match an arrival of type $j$ at $t$ to $i^*$, we know $i^*$ is free at time $t$, because it is in $I_{j,t}$. 

\paragraph{Notation.} Generalizing our notation for the Bernoulli algorithm, we  define $y_{i,t} := \sum_{t' < t} \sum_j x_{i,j,t'}$. 

\medskip

We now prove our upper bound on the probability a subset of offline nodes $I$ are all occupied, proceeding as in \cite{braverman2022max}. 
\begin{lem}\label{lem:NCDcoupling}
    For any $t$ and subset of offline nodes $I$ we have $$\Pr \left[ \bigwedge_{i \in I} \overline{F_{i,t}} \right] \le \prod_{i \in I} y_{i,t}.$$
\end{lem}
\begin{proof}
We proceed by induction on $t$, with the base case being trivial.
For the inductive step, using that at most one node in $I$ is matched to $t$, we have the following.
\begin{align*}
    \Pr \left[ \bigwedge_{i \in I} \overline{F_{i,t+1}} \right] &\le \Pr[F_t \cap I = \emptyset ] + \sum_{i \in I} \Pr[F_t \cap I = \{i\}  ] \cdot \sum_j p_{j,t} \cdot r_{i,j,t} && \text{Property~\ref{level-set:marginals}} \\
    &\le \sum_{H \subseteq I} \Pr \left[F_t \cap I = H \right] \cdot \prod_{i \in H} \left( \sum_{j} \frac{x_{i, j, t}}{1 - y_{i,t}} \right) \\
    &= \sum_{H \subseteq I} \Pr[F_t \cap I \subseteq H] \cdot \prod_{i \in H} \left(   \sum_{j} \frac{x_{i, j, t}}{1 - y_{i,t}} \right) \cdot \prod_{i \in I \setminus H} \left(  1 - \sum_{j} \frac{x_{i, j, t}}{1 - y_{i,t}} \right) && \text{\Cref{fact:independent-buckets}} \\
    &\le \sum_{H \subseteq I} \prod_{i \in I \setminus H} \left[ y_{i,t} \left( 1 - \sum_{j} \frac{x_{i, j, t}}{1 - y_{i,t}} \right) \right] \prod_{i \in H}  \left(   \sum_{j} \frac{x_{i, j, t}}{1 - y_{i,t}} \right) && \text{I.H.} \\
    &= \prod_{i \in I} \left(  y_{i,t} \left( 1 - \sum_{j} \frac{x_{i, j, t}}{1 - y_{i,t}} \right) + \sum_j \frac{x_{i,j,t}}{1-y_{i,t}} \right) && \hspace{-3em} \textrm{Multi-binomial Thm} \\
    &= \prod_{i \in I} y_{i,t}+\sum_j x_{i,j,t} 
    \\
    &= \prod_{i \in I} y_{i,t+1}. && \qedhere
\end{align*}
\end{proof}

The analysis of our algorithm for the Bernoulli case can be extended syntactically to general distributions, by incorporating an additional index. In particular we define an \emph{extended type} to be an ordered pair $(j,t)$. Scaling as in \Cref{sec:scaling} generalizes as follows.

\begin{Def}\label{def:scaling-gen-app}
For non-decreasing function $f: [0,1] \mapsto \mathbb{R}_{\ge 0}$ with $\int_0^1 f(z) \, dz = 1$, define for~$i,j,t$:
\begin{align}
\hat{x}_{i,j,t} & := \int_{y_{i,t}}^{y_{i,t}+x_{i,j,t}} f(z) \, dz,
\\
\hat{y}_{i,t} & := \sum_{t' < t} \sum_j \hat{x}_{i,j,t'} \leq \int_{0}^{y_{i,t}} f(z) \, dz. \label{eqn:yithat-gen}
\end{align} 
\end{Def}

Note the inequality in \Cref{eqn:yithat-gen} (which is an equality for the Bernoulli problem) follows from monotonicity of $f(\cdot)$.

Again, the transformation $\vec{x}\mapsto \vec{\hat{x}}$ preserves Constraint \eqref{eqn:PPSWConstraintgen}, i.e.,  $\hat{r}_{i,j,t}:=\frac{\hat{x}_{i,j,t}}{p_t(1-\hat{y}_{i,t})}$ is in $[0,1]$ if $r_{i,j,t}\in [0,1]$. Non-negativity is trivial, while the upper bound is proven in the following. 
\begin{claim}
    If $\vec{x}$ satisfies Constraint \eqref{eqn:PPSWConstraintgen}, then $\hat{x}_{i,j,t} \le p_{j,t} \cdot \left( 1 - \hat{y}_{i,t} \right)$ for all $i,j,t$.
\end{claim}
\begin{proof}
    Using the definition of $\hat{x}_{i,j,t}$, monotonicity of $f(\cdot)$, Constraint \eqref{eqn:PPSWConstraintgen} and $\int_{0}^1 f(z) \; dz = 1$, and finally that $\hat{y}_{i,j,t}\leq \int_{0}^{y_{i,t}} f(z)\; dz$, we obtain our desired bound. (The only change compared to the proof of \Cref{hatxwelldefined} is the final inequality, which is an equality for the former proof.)
\begin{align*}
    \hat{x}_{i,j,t} &= \int_{y_{i,t}}^{y_{i,t} + x_{i,j,t}} f(z) \, dz  \le \frac{x_{i,j,t}}{1-y_{i,t}} \cdot \int_{y_{i,t}}^1 f(z) \, dz 
    \le p_{j,t} \cdot \left( 1 - \int_0^{y_{i,t}} f(z) \, dz \right) \leq p_{j,t} \cdot (1 - \hat{y}_{i,t} ).  \qedhere
    \end{align*}
\end{proof}

\begin{lem}
    For any fixed $w \ge 0$ define $R_{j,t,w} := \sum_{i: w_{i,j,t} \ge w} r_{i,j,t} \cdot F_{i,t}.$ Then, if $\mathcal{M}$ denotes the matching produced by running \Cref{alg:proposals-core-generalized} on $\{\hat{x}_{i, j,t}\}$, we have $$\E[w(\mathcal{M})] \ge \sum_{j,t} p_{j,t} \cdot \int_0^{\infty} \E[\min(1,\hat{R}_{j,t,w})] \, dw. $$
\end{lem}
\begin{proof}[Proof]
This follows by Property~\ref{level-set:prefix} of the invocation of $\mathsf{PS}(\cdot, \cdot)$ in Line~\ref{line:calltoSRgen}. 
\end{proof}
Thus, following \Cref{sec:edge-weighted}, it suffices to show $$\E[\min(1,\hat{R}_{j,t,w})] \ge \edgeapprox \cdot \sum_{i: w_{i,t} \ge w} \frac{x_{i,j,t}}{p_{j,t}}  .$$
Now, this bound follows by a rather syntactic generalization from the analysis of that section, provided we establish the desired bounds on $\E[\min(1,X)]$ for $X$ not necessarily NCD.

It remains to note that our tail expectation bounds, i.e., our lower bound for $\E[\min(1,X)]$ of \Cref{sec:Emin1X}, are still relevant when analyzing this algorithm, even though the $\{F_{i,t}\}_i$ are not NCD in this generalization. However, our upper bound on $\Pr[\wedge_{i \in I} \overline{F_{i,t}}]\leq\prod_{i\in I} y_{i,t}$ of \Cref{lem:NCDcoupling} is still sufficient for the (fractional) bucketing bound to hold.
\begin{lem}
    Let $X:=\sum_{i=1}^n c_i\cdot X_i$, with $c_i\in [0,1]$ and $X_i\sim \textup{Ber}(q'_i)$ for all $i\in [n]$ and with $\Pr[\bigwedge_{i\in I} \overline{X_i}]\leq \prod_{i\in I} (1-q_i)$ for all $I\subset [n]$. Then,
    \begin{align}\E[\min(1,X)] & \geq 1-\prod_i(1-c_i\cdot q_i), \label{eqn:indcoin-generalization}\\
    \E[\min(1,X)] & \ge 1 - \left( 1 - (1-\theta) \cdot \Big\{ \frac{\mu_S}{1-\theta} \Big\} \right) \cdot \theta^{\lfloor \frac{\mu_S}{1-\theta} \rfloor} \cdot \prod_{i \notin S} \left( 1 - c_i \cdot q_i \right). \label{eqn:fracbucket-generalization}
    \end{align}
\end{lem}
\begin{proof}[Proof]
For \Cref{eqn:indcoin-generalization}, we can directly confirm that the proof of \Cref{lem:basic-independent-proposal-bound} still goes through in this more general setting; 
the only change to the proof is the justification of \eqref{eqn:crucial-NCD} by the current lemma's assumption that $\Pr[ \wedge_{i \in I} \overline{X_i}] \le 1-q_i$. (NCD is a sufficient condition for this condition, but is of coure not necessary for it.)

For \Cref{eqn:fracbucket-generalization}, we can likewise confirm that the proof of \Cref{lem:fracbucketingbound} goes through, as the generalized pivotal sampling step of \Cref{lem:pipage}
holds regardless of the correlations between $\{X_i\}$, after which the argument in \Cref{eqn:fracbucket-generalization} only requires an application of \Cref{eqn:indcoin-generalization}, which we already established holds for $\{X_i\}$ as in the current lemma.
\end{proof}

Finally, the proof of \Cref{lem:edgeweightedbound} can be repeated for each $j$: we split the flow into flow from low-degree and high-degree neighbors, $x_{L, j} := \sum_{i : w_{i,t} \ge w, y_{i,t} \le \theta} \frac{x_{i,j,t}}{p_{j,t}}$ and $x_{H, j} := \sum_{i : w_{i,t} \ge w, y_{i,t} > \theta} \frac{x_{i,j,t}}{p_{j,t}}$, 
and then leveraging \Cref{eqn:fracbucket-generalization}, we obtain that for $\eps=0.11$ and $\delta=0.18$ in the step function $f(\cdot)$ as in \Cref{def:step}, our algorithm satisfies $$\E[\min(1,\hat{R}_{j,t,w})] \ge \edgeapprox \cdot \sum_{i: w_{i,t} \ge w} \frac{x_{i,j,t}}{p_{j,t}}.$$
By the same arguments as in \Cref{thm:edge-weighted}, this then implies our result for general distributions.

\begin{thm}\label{thm:edge-weighted-generalization}
    \Cref{alg:proposals-core-generalized} with input $\hat{\vec{x}}$, computed as in Definitions \ref{def:scaling-gen-app} and \ref{def:step}  (with $\eps=0.11$ and $\delta=0.18$) from $\vec{x}$ an optimal solution to \eqref{LP-PPSW-Gen}, is a polynomial-time $\edgeapprox$-approximate algorithm for edge-weighted stochastic online matching.
\end{thm}

%% file: app-pricing.tex
\section{A Truthful Mechanism for Matching Markets} \label{app-pricing}

In this section, we show our algorithm has implications for Bayesian mechanism design in matching markets. Consider a setting $n$ offline items $I$ and $T$ buyers who have unit-demand valuations over these items.\footnote{We recall a \emph{unit-demand valuation} $v(\cdot)$ is one such that $v(S) = \max_{T \subseteq S, |T| \le 1} v(T)$.} Each buyer $t \in [T]$ samples his valuation function $v_t(\cdot)$ from a known distribution $\mathcal{D}_t$; at time $t$ we must irrevocably decide which item to assign to $t$. Our goal is to maximize social welfare, defined as the sum of the valuations of all agents. We refer to this as an \emph{online Bayesian matching market}. 

Note that if the realization of each valuation is revealed to us, this is equivalent to the setting considered earlier, and \Cref{alg:proposals-core-generalized} will in polynomial-time achieve a $\edgeapprox$-approximation to the social welfare achievable by the optimum online algorithm. In this section, we argue that this guarantee extends to the case where buyers are strategic agents, and must not truthfully reveal their valuation. There is an extensive line of work studying this setting, and generalizations thereof, showing that we can achieve a (tight) $\nicefrac{1}{2}$-approximation to the social welfare of the allocation of optimum offline. This work is the first to explicitly study the question against the online benchmark. 

We achieve our guarantee by using an (adaptive) pricing-based mechanism, which is easily observed to be dominant-strategy incentive-compatible (DSIC). The prices are set via a recent result of \cite{banihashem2024power}, which guarantees there is no decrease in the social welfare. We can further show that these prices are computable efficiently. 

\begin{definition}
    A matching mechanism is \emph{pricing-based} if for every buyer $t$, before the arrival of $t$ it sets a price $\pi_i$ for every item $i$, depending only on the set of items remaining and $\mathcal{D}_t$ (not its realization). It then queries buyer $t$ for his values $v_t(i)$ for each item $i \in I$, assigns buyer $t$ the item given by $i^* := \arg\max_{i \in \{\emptyset \} \cup I}  (v_t(i) - \pi_i)$, and charges him a price of $\pi_{i^*}$. (We naturally define $v_t(\emptyset)$ and $\pi_{\emptyset}$ to be 0.)
\end{definition}

\begin{obs}
    A pricing-based mechanism is DSIC. 
\end{obs}

Using \cite{banihashem2024power}, we will show that \Cref{alg:proposals-core-generalized} can be converted into one that is pricing-based with (i) no loss in the expected social welfare and (ii) prices that are computable in polynomial time. 

\begin{lemma}
    There exists a pricing-based mechanism $\mathcal{M}$ for online Bayesian matching markets obtaining social welfare is at least a $\edgeapprox$-approximation to that of the optimum online algorithm, which sees the true valuations. Furthermore, each computation of prices by $\mathcal{M}$ can be done in polynomial-time. 
\end{lemma}

\begin{proof}
    We rely on \cite[Theorem 4]{banihashem2024power}. Our setting of an online Bayesian matching market is a special case of their model where the \emph{outcome space} of each agent is the subsets of items of size at most 1, and the \emph{feasible outcomes} are those that allocate each item at most once (i.e., form a matching). We note that in our setting the arrival order of online agents is known, unlike that considered by \cite{banihashem2024power}, but the difference can be handled syntactically. By \cite[Theorem 4]{banihashem2024power}, there hence exists a pricing-based mechanism $\mathcal{M}$ which achieves social welfare at least that of \Cref{alg:proposals-core-generalized} in expectation, and furthermore maintains the exact distribution over assigned outcomes (i.e., the joint distribution of $\{F_{i,t}\}_i$ for each $t$ is unchanged). 
    
   It remains to argue that the prices of $\mathcal{M}$ can be computed efficiently. Note first that \Cref{alg:proposals-core-generalized} is \emph{past-valuation-independent}, in the sense that its matching decision for agent $t$ depends only on the set of free items $\{F_{i,t}\}_i$, the arriving valuation $v_t(\cdot)$, and the input distributions. This means, that conditioned on a fixed realization $R$ of $\{F_{i,t}\}_i$, we can efficiently calculate the conditional probabilities $$x_{i,t}^R := \Pr[(i,t) \in \mathcal{M} \mid \{ F_{i,t} \}_i = R]$$ and $$x_{\emptyset}^R := \Pr[ t \not\in \mathcal{M} \mid \{ F_{i,t} \}_i = R].$$ Thus for any fixed $R$, the linear program \cite[(LP1)]{banihashem2024power} has at most $|\mathcal{D}_t| \cdot (n+1)$ variables, and coefficients that can be computed efficiently. Hence we can solve this LP in polynomial time. As the reduction of \cite[Theorem 4]{banihashem2024power} only requires solving this LP at most $t$ times (for each observed realization of $\{F_{i,t}\}_i$), we hence can compute all prices used in one run of the mechanism. 
\end{proof}